\documentclass[11pt]{article}

\usepackage[svgnames]{xcolor} 
\usepackage{authblk}

\usepackage[utf8]{inputenc}
\usepackage[T1]{fontenc}

\usepackage{amsmath}
\usepackage{amssymb}
\usepackage{amsthm}

\usepackage{bussproofs}

\newtheorem{theorem}{Theorem}[section]
\newtheorem{lemma}[theorem]{Lemma}
\newtheorem{corollary}[theorem]{Corollary}
\newtheorem{proposition}[theorem]{Proposition}
\theoremstyle{definition}
\newtheorem{definition}[theorem]{Definition}

\date{\today}
\title{Algebraic Proof Theory for Infinitary Action Logic}

\newcommand\under{\backslash}
\newcommand\ovr{/}
\newcommand\lresidual{\backslash}
\newcommand\rresidual{/}
\newcommand\Lresidual{\backslash\kern-0.5ex\backslash}
\newcommand\Rresidual{/ \kern-0.5ex /}
\newcommand\NN{\mathbb{N}}
\newcommand\set[1]{\left\{#1\right\}}

\newcommand\AL{\mathsf{AL}}
\newcommand\ALstar{\mathsf{AL}^*}
\newcommand\gentzenContinuousActionLattice{\mathrm{G}^\infty\mathsf{ACT}^*}
\newcommand\gentzenContinuousActionLatticeOmega{\mathrm{G}^\omega\mathsf{ACT}^*}
\DeclareMathOperator\cut{\mathrm{Cut}}

\newcommand\dualAlgebra{\mathbf{W}^+}
\DeclareMathOperator\term{Fm}
\DeclareMathOperator\termAlg{\mathbf{Fm}}
\DeclareMathOperator\Seq{\mathbf{Seq}}

\DeclareMathOperator\tail{\mathsf{tl}}

\DeclareMathOperator\height{hg}
\DeclareMathOperator\length{lg}

\DeclareMathOperator\branches{Br}

\DeclareMathOperator\Projection{Pr}
 
\DeclareMathOperator\Om{Om} 
\newcommand\nwProof{\mathbb{P}^\infty}
\newcommand\wProof{\mathbb{P}^\omega}
\newcommand\restricts{\mathord{\upharpoonright}}

\DeclareMathOperator\path{P}
\DeclareMathOperator\criticalHeight{chg}
\DeclareMathOperator\progressPoints{\mathrm{PP}}

\DeclareMathOperator\nodes{Node}

\DeclareMathOperator\domain{Dom}
\DeclareMathOperator\image{Im}
\newcommand{\Var}{\mathrm{Var}}
\newcommand{\FVar}{\mathrm{FVar}}
\newcommand{\SVar}{\mathrm{SVar}}
\newcommand{\m}{\mathbf}
\newcommand{\meet}{\wedge}
\newcommand{\join}{\vee}
\DeclareMathOperator\mwffAlg{\mathbf{MFm}}

\author[1]{Wesley Fussner}
\author[2]{Simon Santschi}
\author[3]{Borja Sierra Miranda}
\affil[1]{Institute of Computer Science of the Czech Academy of Sciences, Czechia}
\affil[2]{Institute of Mathematics, University of Bern, Switzerland}
\affil[3]{Institute of Informatics, University of Bern, Switzerland}

\begin{document}

\maketitle

\begin{abstract}
We exhibit a uniform method for obtaining (wellfounded and non-wellfounded) cut-free sequent-style proof systems that are sound and complete for various classes of action algebras, i.e., Kleene algebras enriched with meets and residuals. Our method applies to any class of $*$-continuous action algebras that is defined, relative to the class of all $*$-continuous action algebras, by analytic quasiequations. The latter make up an expansive class of conditions encompassing the algebraic analogues of most well-known structural rules. These results are achieved by wedding existing work on non-wellfounded proof theory for action algebras with tools from algebraic proof theory.
\end{abstract}

\section{Introduction}

The equational logic of regular expressions has been of widespread interest since Kleene's pioneering work on the subject. Kleene algebras were first axiomatized by Kozen in order to provide a finitary presentation of this theory, following decades of work by J.H. Conway, V.N. Redko, and others; see \cite{Redko1964,C1971,Koz1991,Koz1994a}. \emph{Action lattices} \cite{P1990,Koz1994b,J2004} extend Kleene algebras by a pair of division-like operations $\under$ and $\ovr$, often called \emph{residuals}, as well as an operation $\wedge$ capturing binary infima. In addition to addressing several theoretical concerns, the inclusion of $\under,\ovr,\wedge$ in the language places action lattices under the same umbrella as substructural logics and linear logic (see \cite{MPT23}), where residuals are used to model implication and $\wedge$ models weak conjunction. 

Against this backdrop, there has recently been substantial work on action lattices with \emph{hypotheses}---i.e., classes of action lattices defined relative to all action lattices by quasiequations; see \cite{DKPP2019,KPS2024}. The present study provides a uniform methodology for obtaining cut-free proof systems for such classes of action lattices with hypotheses, subject to the conditions that (1) the action lattices are $*$-continuous and (2) the defining quasiequations are \emph{analytic} (see Section~\ref{sec:cut-free completeness}). In particular, building on the work of Das and Pous \cite{DP2018}, for each class $\mathsf{K}$ of $*$-continuous action lattices defined by analytic quasiequations, we exhibit two equivalent cut-free proof systems that are each sound and complete for $\mathsf{K}$; see Theorem~\ref{thm:main}. As a consequence, we prove that each class of $*$-continuous action lattices defined by analytic quasiequations is closed under MacNeille completions; see Corollary~\ref{cor:MacNeille}.

The aforementioned results are obtained by a combination of proof theoretic and algebraic methods. In Section~\ref{sec:proof systems}, we introduce the two previously mentioned classes of proof systems---one comprising a class of wellfounded systems and the other a class of non-wellfounded proof systems---and further show that these systems are mutually interpretable. The translation between the two systems is an essential ingredient to showing soundness for the non-wellfounded system. The completeness of the systems is subsequently established in Section~\ref{sec:cut-free completeness}, and draws on recent developments in algebraic proof theory, first developed in the context of substructural logics; see \cite{CGT2012,CGT2017}. We begin by rehashing some needed preliminaries in Section~\ref{sec:preliminaries}.

\section{Preliminaries}\label{sec:preliminaries}

\subsection{Action lattices}

An algebra \((A, \vee,  \cdot, {^*}, 0,1)\) is called a  \emph{Kleene algebra} if  \((A, \vee, \cdot, 0, 1)\)  is an (additively) idempotent semiring and for all  $x,y\in A$,
	\begin{align*}
			&1 \vee xx^* \leq x^*, \\
			&xy \leq y \Longrightarrow x^*y \leq y,\\
			&yx \leq y \Longrightarrow yx^* \leq y,
		\end{align*}

An algebra $(A,\wedge, \vee, \cdot, \lresidual, \rresidual, 1 )$ is called a \emph{residuated lattice} if  $(A,\wedge, \vee)$ is a lattice, $(A,\cdot, 1)$ is a monoid,  and  for all $x,y,z\in A$,
	\[
	x\cdot y\leq z \iff y\leq x\lresidual z\iff x\leq z\rresidual y,
	\]
where $\leq$ is the lattice order.

	This paper's main objects of study are action lattices, which can be seen both as extensions of residuated lattices with a Kleene star and a least element, and as extensions of Kleene algebras with a binary meet and residual for the multiplication. More formally, an \emph{action lattice} is an algebraic structures of the form 
	 \((A, \wedge, \vee, \cdot, \lresidual,\rresidual, {^*}, 0,1)\), where:
	 \begin{enumerate}
	 \item $(A,\wedge, \vee, \cdot, \lresidual, \rresidual, 1 )$ is a residuated lattice,
	 \item for each $x\in A$, $0\leq x$,
	 \item and for all $x,y\in A$,
	 	\begin{align*}
			&1 \vee xx^* \leq x^*, \\
			&xy \leq y \Longrightarrow x^*y \leq y,\\
			&yx \leq y \Longrightarrow yx^* \leq y,
		\end{align*}
	 \end{enumerate}
	Kleene algebras have been studied extensively as algebras of regular expressions, but are shown in \cite{Redko1964} to comprise a proper quasivariety (i.e., the class of all Kleene algebras cannot be defined by purely equational conditions). In contrast, the class $\AL$ of all action lattices is a finitely based equational class \cite{P1990}. Importantly and like the class of all Kleene algebras, $\AL$ is also closed under the formation of matrices; the same is not true if $\wedge$ is excluded from the language \cite{Koz1994b}.
	
	Most of the important examples of Kleene algebras (and action lattices) are \emph{\({*}\)-continuous} in that for any
	\(a \in A\) we have
	\[ 
		a^* = \bigvee \set{a^n \mid n \in \NN}, 
	\]
	where \(a^0 := 1\) and \(a^{n+1} := a^n \cdot a\) and $\bigvee$ denotes (possibly infinitary) supremum.  
Pratt's normality theorem \cite{P1990} guarantees that any action lattice that is \emph{complete} (in the sense that arbitrary suprema exist) is ${*}$-continuous. We denote the class of ${*}$-continuous action lattices by $\ALstar$.

\subsection{Non-wellfounded trees and corecursion}

Later, we will require some notioins for working with non-wellfounded proof systems. We fix some terminology. A \emph{sequence}  \emph{\(s\)  on \(A\)} is a function whose domain, called the \emph{length} of \(s\) and denoted \(\length(s)\), is either a natural number or \(\omega\).
The \emph{tail} of a sequence \(s\), denoted \(\tail(s)\), is the sequence obtained after removing its first element.
For \(a \in A\), \(a : s\) denotes the sequence obtained from \(s\) by adding \(a\) to the start.

A \emph{tree (with labels in \(A\))} is a function \(T\) such that
\begin{enumerate}
  \item \(\domain(T) \subseteq \mathbb{N}^{<\omega}\) and \(\image(T) \subseteq A\).
  \item \(\domain(T)\) is closed under prefixes.
\end{enumerate}
The elements of \(\domain(T)\) are called \emph{nodes of \(T\)}.

Given a tree \(T\) an \emph{(infinite) branch} is an infinite sequence \((w_i)_{i \in \mathbb{N}}\) of nodes of \(T\) such that \(w_0 = \epsilon\) and \(w_{i+1} = w_i n\) for some \(n \in \mathbb{N}\).
The set of branches of \(T\) is denoted as \(\branches(T)\).

A tree is said to be
\begin{enumerate}
  \item \emph{Finite} if \(\domain(T)\) is finite.
  \item \emph{Finitely branching} if for any \(w \in \domain(T)\) the set \(\set{n \in \mathbb{N} \mid wn \in \domain(T)}\) is bounded.
  \item \emph{Wellfounded} if \(\branches(T) = \varnothing\).
\end{enumerate}

We say that \(T'\) is a \emph{subtree} of a tree \(T\) if there exists a \(w \in \nodes(T)\) such that
\begin{enumerate}
  \item \(w' \in \domain T' \iff ww' \in \domain T \), and
  \item \(T'(w') = T(ww')\).
\end{enumerate}
In this case, \(T'\) is also called the subtree of \(T\) \emph{generated at node} \(w\).
A subtree \(T'\) of a tree \(T\) is called \emph{proper} iff it is not equal to \(T\).

Any wellfounded tree \(T\) can be assigned an ordinal, the \emph{height of \(T\)} denoted \(\height{T}\), in such a way that any proper subtree \(T'\) of \(T\) fulfills \(\height(T') < \height(T)\).

\section{Two proof systems and their equivalence}\label{sec:proof systems}

In this section, we give a pair of proof systems that are sound and complete for $*$-continuous action lattices---a wellfounded systems \(\gentzenContinuousActionLatticeOmega\), and a non-wellfounded system \(\gentzenContinuousActionLattice\) due to Das and Pous \cite{DP2018}--- and show that the two systems are mutually interpretable. Extensions of each of these systems by analytic structural rules (see Definition~\ref{def:struct-rules}) may readily be seen---using the algebraic proof theory discussed in Section~\ref{sec:cut-free completeness}---to be complete with respect to corresponding subclasses of $\AL^*$. The soundness of the non-wellfounded system with respect to the aforementioned classes is a more difficult problem, however, and we prove it by appealing to the back-and-forth translation described in the present section. Whether this soundness result may be proven directly for the non-wellfounded system, without appealing to the translation, remains open.

We first set some notation. Fix a countably infinite set $\Var$ of propositional variables and write $\termAlg$ for the absolutely free algebra generated by $\Var$ over the language
$\{\wedge, \vee, \cdot, \lresidual, \rresidual,  {^*}, 0,1 \}$. We call an element of the algebra $\termAlg$ a \emph{formula} and sometimes refer to $\termAlg$ as the \emph{formula algebra} defined over $\Var$. A \emph{sequent} is a pair $(\Gamma, \beta)$ where $\Gamma$ is a finite sequence
of formulas and $\beta$ is a formula. We write $\Gamma \Rightarrow \beta$
for the sequent $(\Gamma, \beta)$, and we denote by $\Seq$ the collection of all sequents. We will write $\Gamma,\Sigma$ for the concatenation of two sequences $\Gamma$ and $\Sigma$ and for a formula $\alpha$ we define recursively the sequence $\alpha^{(n)}$ by $\alpha^{(0)} = \epsilon$, $\alpha^{(k+1)} =
\alpha^{(k)},\alpha$.

\begin{definition}
	Let \(S = \Pi \Rightarrow \beta\) be a sequent.
	We define \(|S| = \length(\Pi)\) and for \(i \in |S| \cup \set{-1}\) we will
	write \(S(i)\) to mean:
	\begin{align*}
		&S(-1) := \beta, \\
		&S(i) := \Pi(i), \text{ if } i \geq 0.
	\end{align*}
	Here \(\Pi(i)\) is the \(i\)-th element in the sequence \(\Pi\),
	starting to count at \(0\).
	The elements of \(|S| \cup \set{-1}\) are called the 
	\emph{formula occurrences of \(S\)} and we say that \(i\) is an
  \emph{occurrence of a formula \(\alpha\)} if \(S(i) = \alpha\).
\end{definition}

\subsection{Rules}

Fix two disjoint countably infinite sets \(\FVar\), \(\SVar\), whose elements are respectively 
called \emph{formula metavariables} and \emph{sequence metavariables} (we also assume
that these sets are disjoint from the set of variables).
We write \(\mwffAlg\) for the absolutely free algebra generated by \(\FVar\) over the language \(\set{\meet,\join, \cdot,\lresidual, \rresidual,{^*},0,1}\).
A \emph{metaformula} is an element of \(\mwffAlg\).

\begin{definition}
  A subset $R$ of $\Seq^n$ is called an \emph{\(n\)-ary rule} and we say that $R$ is a \emph{rule} if 
  it is an \(n\)-ary rule for some natural number $n$. If $R$ is any rule, each element of $R$ is called
  an \emph{instance} of the rule $R$.

	An \emph{$n$-ary schematic rule} is a tuple
	$(\mathcal{S}_0, \ldots, \mathcal{S}_n)$ where each \(\mathcal{S}_i\) is an
	ordered pair \((\Upsilon_i, \phi_i)\) such that
	\begin{enumerate}
		\item \(\Upsilon_i\) is a sequence of sequence metavariables and
		metaformulas.
		\item \(\phi_i\) is a metaformula.
	\end{enumerate}
	A \emph{metavariable instantiation} is a function that maps sequence
	metavariables to sequences of formulas and formula metavariables to formulas.
	Any instantiation can be extended, thanks to the free property of absolutely free algebras, in a unique way to
  a function that maps sequence metavariable to sequences of formulas and metaformulas to formulas.\footnote{We will identify an instantiation with its extension.} Note that any \(n\)-ary schematic rule defines an \(n\)-ary rule by taking the sets of all its instantiations.\footnote{We will identify an schematic rule with the rule it defines.}
\end{definition}

\begin{figure}
 
	\[
	\AxiomC{}
	\RightLabel{\(0\)L}
  \UnaryInfC{\(\Gamma, \fbox{0}, \Delta \Rightarrow \beta\)}
	\DisplayProof
	\hspace{0.5cm}
	\AxiomC{\(\Gamma, \Delta \Rightarrow \beta\)}
	\RightLabel{\(1\)L}
  \UnaryInfC{\(\Gamma, \fbox{1}, \Delta \Rightarrow \beta\)}
	\DisplayProof
  \]

  \[
	\AxiomC{}
	\RightLabel{\(1\)R}
  \UnaryInfC{\(\Rightarrow \fbox{1}\)}
	\DisplayProof
	\]

	\[ 
    \AxiomC{\(\Gamma, \fbox{\(\alpha_i\)}, \Delta \Rightarrow \beta\)}
		\RightLabel{\(\wedge\)L\(_i\)}
    \UnaryInfC{\(\Gamma, \fbox{\(\alpha_0 \wedge \alpha_1\)}, \Delta \Rightarrow \beta\)}
		\DisplayProof
	\]

  \[
    \AxiomC{\(\Gamma \Rightarrow \fbox{\(\beta_0\)}\)}
    \AxiomC{\(\Gamma \Rightarrow \fbox{\(\beta_1\)}\)}
		\RightLabel{\(\wedge\)R}
    \BinaryInfC{\(\Gamma \Rightarrow \fbox{\(\beta_0 \wedge \beta_1\)}\)}
		\DisplayProof
  \]

	\[ 
    \AxiomC{\(\Gamma, \fbox{\(\alpha_0\)}, \Delta \Rightarrow \beta\)}
    \AxiomC{\(\Gamma, \fbox{\(\alpha_1\)}, \Delta \Rightarrow \beta\)}
		\RightLabel{\(\vee\)L}
    \BinaryInfC{\(\Gamma, \fbox{\(\alpha_0 \vee \alpha_1\)}, \Delta \Rightarrow \beta\)}
		\DisplayProof
	\]

  \[
    \AxiomC{\(\Gamma \Rightarrow \fbox{\(\beta_i\)}\)}
		\RightLabel{\(\vee\)R\(_i\)}
    \UnaryInfC{\(\Gamma \Rightarrow \fbox{\(\beta_0 \vee \beta_1\)}\)}
		\DisplayProof
  \]

	\[ 
    \AxiomC{\(\Gamma, \fbox{\(\alpha_0\)}, \fbox{\(\alpha_1\)}, \Delta \Rightarrow \beta\)}
		\RightLabel{\(\cdot\)L}
    \UnaryInfC{\(\Gamma, \fbox{\(\alpha_0 \cdot \alpha_1\)}, \Delta \Rightarrow \beta\)}
		\DisplayProof
	\]

  \[
    \AxiomC{\(\Gamma \Rightarrow \fbox{\(\beta_0\)}\)}
    \AxiomC{\(\Delta \Rightarrow \fbox{\(\beta_1\)}\)}
		\RightLabel{\(\cdot\)R}
    \BinaryInfC{\(\Gamma, \Delta \Rightarrow \fbox{\(\beta_0 \cdot \beta_1\)}\)}
		\DisplayProof
  \]

	\[ 
    \AxiomC{\(\Delta \Rightarrow \fbox{\(\alpha_0\)}\)}
    \AxiomC{\(\Gamma, \fbox{\(\alpha_1\)}, \Sigma \Rightarrow \beta \)}
		\RightLabel{\(\lresidual\)L}
    \BinaryInfC{\(\Gamma, \Delta, \fbox{\(\alpha_0 \lresidual \alpha_1\)}, \Sigma \Rightarrow \beta\)}
		\DisplayProof
	\]

	\[ 
    \AxiomC{\(\Delta \Rightarrow \fbox{\(\alpha_0\)}\)}
    \AxiomC{\(\Gamma, \fbox{\(\alpha_1\)}, \Sigma \Rightarrow \beta\)}
		\RightLabel{\(\rresidual\)L}
    \BinaryInfC{\(\Gamma, \fbox{\(\alpha_1 \rresidual \alpha_0\)}, \Delta, \Sigma \Rightarrow \beta\)}
		\DisplayProof
	\]

  \[
    \AxiomC{\(\fbox{\(\beta_0\)}, \Gamma \Rightarrow \fbox{\(\beta_1\)}\)}
		\RightLabel{\(\lresidual\)R}
    \UnaryInfC{\(\Gamma \Rightarrow \fbox{\(\beta_0 \lresidual \beta_1\)}\)}
		\DisplayProof
		\hspace{0.5cm}
    \AxiomC{\(\Gamma, \fbox{\(\beta_0\)} \Rightarrow \fbox{\(\beta_1\)}\)}
		\RightLabel{\(\rresidual\)R}
    \UnaryInfC{\(\Gamma \Rightarrow \fbox{\(\beta_1 \rresidual \beta_0\)}\)}
		\DisplayProof
  \]

	\[
		\AxiomC{}
		\RightLabel{\(*\)R\(_0\)}
    \UnaryInfC{\(\Rightarrow \fbox{\(\beta^*\)}\)}
		\DisplayProof
		\hspace{0.5cm}
    \AxiomC{\(\Gamma \Rightarrow \fbox{\(\beta\)}\)}
    \AxiomC{\(\Delta \Rightarrow \fbox{\(\beta^*\)}\)}
		\RightLabel{\(*\)R\(_1\)}
    \BinaryInfC{\(\Gamma,\Delta \Rightarrow \fbox{\(\beta^*\)}\)}
		\DisplayProof
	\]

  \[
    \AxiomC{\((\Gamma,\fbox{\(\alpha^{(n)}\)},\Delta \Rightarrow \beta)_{n\in \NN}\)}
    \RightLabel{\(*\)L\(_\omega\)}
    \UnaryInfC{\(\Gamma, \fbox{\(\alpha^*\)}, \Delta \Rightarrow \beta\)}
    \DisplayProof
  \]

  \[
    \AxiomC{\(\Gamma, \Delta \Rightarrow \beta\)}
    \AxiomC{\(\Gamma, \fbox{\(\alpha\)}, \fbox{\(\alpha^*\)}, \Delta \Rightarrow \beta\)}
    \RightLabel{\(*\)L}
    \BinaryInfC{\(\Gamma, \fbox{\(\alpha^*\)}, \Delta \Rightarrow \beta\)}
    \DisplayProof
  \]
	\caption{Principal rules}
	\label{fig:common-rules}\end{figure}

We define an important rule called the \emph{identity axiom} as the set of all instances
\[
  \AxiomC{}
  \RightLabel{\text{id}}
  \UnaryInfC{\(a \Rightarrow a\)}
  \DisplayProof
\]
where \(a\) is a propositional variable.
We note that this rule is not schematic.

In Figure~\ref{fig:common-rules}, we define some schematic rules, which we will call \emph{principal rules}.
In these rules, \(\Gamma, \Delta, \Sigma\) are sequence metavariables
and \(\alpha, \beta\) (possibly with subscripts) are formula metavariables.
In each instance of a principal rule there is one formula occurrence at the conclusion which we call \emph{principal} (this is the boxed formula in the conclusion),
and some formula ocurrences in the premises which we call \emph{auxiliary} (these are the boxed formulas in the premises). The rest of formulas ocurrences will be said to belong to the \emph{context}.

We now introduce some additional sets of rules.
\begin{definition}\label{def:struct-rules}
  A structural rule is a schematic rule \(R\) of shape
  \[ 
    \AxiomC{\(\Upsilon_0 \Rightarrow \alpha_0\)}
    \AxiomC{\(\ldots\)}
    \AxiomC{\(\Upsilon_{k-1} \Rightarrow \alpha_{k-1}\)}
    \RightLabel{\(R\)}
    \TrinaryInfC{\(\Upsilon_{k} \Rightarrow \alpha_{k}\)}
    \DisplayProof
  \]
  where
  \begin{enumerate}
    \item Each \(\alpha_i\) is a metavariable for formulas (some may coincide) and
    \item Each \(\Upsilon_i\) is a (possibly empty) sequence of metavariables for formulas and sequences of formulas.
  \end{enumerate}
  We will say that the rule is \emph{linear} if \(\Upsilon_k\) consists of distinct sequence metavariables.
  Finally, we will say that a linear rule is \emph{analytic} if it has shape
  \[
    \AxiomC{\(\Gamma, \Upsilon_0, \Delta \Rightarrow \beta\)}
    \AxiomC{\(\cdots\)}
    \AxiomC{\(\Gamma, \Upsilon_{n-1}, \Delta \Rightarrow \beta\)}
    \RightLabel{\(R\)}
    \TrinaryInfC{\(\Gamma, \Upsilon, \Delta \Rightarrow \beta\)}
    \DisplayProof
  \]
  where \(\Upsilon\) is a sequence of sequence metavariables, \(\Gamma, \Delta\) are sequence metavariables, \(\beta\) is a formula metavariable and each \(\Upsilon_i\) is a sequence of metavariables appearing in \(\Upsilon\) (allowing repetition).
  Note that, since analytic rules are linear, the metavariables in \(\set{\Gamma, \Delta} \cup\Upsilon\) are distinct.
\end{definition}

We stipulate that structural rules have no principal nor auxiliary formulas.
Figure~\ref{fig:structural-rules} shows some examples of structural rules.
Note that \(\cut\), \(\text{C}\) and \(\text{Wk}\) are linear structural rules and \(\text{c}\) is not linear since it has formula metavariables in the left-hand side of the conclusion.
The linear rules \(\text{C}\) and \(\text{Wk}\) from Figure~\ref{fig:structural-rules} are analytic, but \(\text{Cut}\) is not.

\begin{figure}
  \[ 
    \AxiomC{\(\Delta \Rightarrow \alpha\)}
    \AxiomC{\(\Gamma, \alpha, \Pi \Rightarrow \beta\)}
    \RightLabel{Cut}
    \BinaryInfC{\(\Gamma, \Delta, \Pi \Rightarrow \beta\)}
    \DisplayProof
  \]
  \[
    \AxiomC{\(\Gamma, \alpha, \alpha, \Delta \Rightarrow \beta\)}
    \RightLabel{c}
    \UnaryInfC{\(\Gamma, \alpha, \Delta \Rightarrow \beta\)}
    \DisplayProof
    \qquad
    \AxiomC{\(\Gamma, \Pi, \Pi, \Delta \Rightarrow \beta\)}
    \RightLabel{C}
    \UnaryInfC{\(\Gamma, \Pi, \Delta \Rightarrow \beta\)}
    \DisplayProof
  \]
  \[
    \AxiomC{\(\Gamma, \alpha, \beta, \Delta \Rightarrow \gamma\)}
    \RightLabel{e}
    \UnaryInfC{\(\Gamma, \beta, \alpha, \Delta \Rightarrow \gamma\)}
    \DisplayProof
    \qquad
    \AxiomC{\(\Gamma, \Delta \Rightarrow \beta\)}
    \RightLabel{Wk}
    \UnaryInfC{\(\Gamma, \Pi, \Delta \Rightarrow \beta\)}
    \DisplayProof
  \]
  \caption{Structural rules}\label{fig:structural-rules}
\end{figure}

Finally, we introduce a relation which will play a fundamental role in the non-wellfounded proof systems.

\begin{definition}
  Given an instantiation of a (principal or structural) non-axiomatic rule, a formula occurrence \(\phi\)
  at the conclusion, and a formula occurrence \(\phi'\) at some premise we say that
  \(\phi'\) is an \emph{immediate ancestor of \(\phi\)} if either
  \begin{enumerate}
		\item \(\phi'\) is an auxiliary formula and \(\phi\)
		is the principal formula.
    \item Both, \(\phi\) and \(\phi'\), are obtained from instantiating the same
      metavariable for formulas.
    \item Both, \(\phi\) and \(\phi'\), are inside the instantiation of the same
      metavariable for sequences at the same position (in the 
      sequence of formulas being used for the instantiation).
      \qedhere
  \end{enumerate}
\end{definition}

\subsection{Preproofs, proofs, and proof systems}

A preproof of a set of rules \(\mathcal{R}\) can be intuitively understood as a (possibly)
non-wellfounded tree generated by the rules in \(\mathcal{R}\).
More formally, given a set of rules \(\mathcal{R}\) a \emph{preproof}
in \(\mathcal{R}\) is a non-wellfounded tree \(\pi\) such that:
\begin{enumerate}
	\item Every node \(w\) of \(\pi\) is labelled by a sequent \(S^\pi_w\) and a
	rule \(R^\pi_w \in \mathcal{R}\).
	\item If \(w\) has successors \(wi_0, \ldots, wi_{n-1}\)
	(where \(i_0 < \cdots < i_{n-1}\)) then
	\((S^\pi_{wi_0}, \ldots, S^\pi_{wi_{n-1}}, S^\pi_w)\) is an instance of the
	rule \(R^\pi_w\).
  \item If \(w\) has a successor \(wi\), then \(wj\) is also a successor for any \(j < i\).
\end{enumerate}
In Subsection~\ref{subsec:translation} we will temporarily lift the third condition, allowing a node \(w\) to have a successor \(w1\) but no successor \(w0\) in one particular rule.
If \(\pi\) is a preproof and \(w\) is a node of \(\pi\), we will write
\(S^\pi_w\) to denote the sequent at \(w\) in \(\pi\) and \(R^\pi_w\) is the
rule at \(w\) in \(\pi\).
If we omit \(w\) we assume \(w = \epsilon\), i.e., \(w\) is the root of the tree.

A \emph{proof system} \(\mathcal{G}\) is given by:
\begin{enumerate}
	\item A set of rules \(\mathcal{R}\), called the
		\emph{rules of \(\mathcal{G}\)}.
	\item A condition on the branches of the preproofs of \(\mathcal{R}\),
		called the branch condition.
\end{enumerate}
Then, a proof in \(\mathcal{G}\) will be a preproof in the rules of
\(\mathcal{G}\) where every branch satisfies the branch condition.
When we talk about a preproof in \(\mathcal{G}\) we mean a preproof
in the rules of \(\mathcal{G}\).
A wellfounded proof system is a proof system where the branch condition is empty
(so no branches are allowed and proofs are just wellfounded preproofs).

\begin{definition}
Let \(\mathcal{R}\) be a set of linear rules.
We define \(\gentzenContinuousActionLatticeOmega + \mathcal{R}\) as the
wellfounded proof system whose rules are \(\text{id}\), the rules in \(\mathcal{R}\) and all the rules in Figure~\ref{fig:common-rules} except for \(*L\).

Since \(\gentzenContinuousActionLatticeOmega + \mathcal{R}\) is a
wellfounded proof system, its proofs are just wellfounded preproofs, i.e., no
infinite paths are allowed.
\end{definition}

The following lemma will be needed in Section~\ref{sec:cut-free completeness} to show the (cut-free) completeness of the system.

\begin{lemma}\label{l:adm-of-frame-rules}
Let $\mathcal{R}$ be a set of analytic rules.
The following rules are admissible in $\gentzenContinuousActionLatticeOmega +\mathcal{R}$:
	\[
		\AxiomC{}
		\UnaryInfC{\(\alpha \Rightarrow \alpha\)}
		\DisplayProof
		\qquad
		\AxiomC{\(\Gamma \Rightarrow 0\)}	
    \RightLabel{\(0\mathrm{R}\)}
		\UnaryInfC{\(\Sigma_l, \Gamma, \Sigma_r \Rightarrow \beta\)}
		\DisplayProof
	\]
\end{lemma}
\begin{proof}
  To show that \(\gentzenContinuousActionLatticeOmega + \mathcal{R} \vdash \alpha \Rightarrow \alpha\) we need to do an induction on the size of \(\alpha\) and case analysis in the shape of \(\alpha\).
  We provide the proof for the case \(\alpha = \gamma^*\).
  First, for any \(n > 0\) we define the proof \(\tau_n\) as
  \begin{prooftree}
    \AxiomC{}
    \RightLabel{I.H.}
    \UnaryInfC{\(\gamma \Rightarrow \gamma\)}
    \AxiomC{}
    \RightLabel{I.H.}
    \UnaryInfC{\(\gamma \Rightarrow \gamma\)}
    \AxiomC{}
    \RightLabel{I.H.}
    \UnaryInfC{\(\gamma \Rightarrow \gamma\)}
    \AxiomC{}
    \RightLabel{*\(\mathrm{R}_0\)}
    \UnaryInfC{\(\Rightarrow \gamma^*\)}
    \RightLabel{\(*\mathrm{R}_1\)}
    \BinaryInfC{\(\gamma^{(1)} \Rightarrow \gamma^*\)}
    \RightLabel{\(*\mathrm{R}_1\)}
    \BinaryInfC{\(\underset{\raise .5em\vdots}{\gamma^{(2)} \Rightarrow \gamma^*}\)}
    \RightLabel{\(*\mathrm{R}_1\)}
    \UnaryInfC{\(\gamma^{(n-1)} \Rightarrow \gamma^*\)}
    \RightLabel{\(*\mathrm{R}_1\)}
    \BinaryInfC{\(\gamma^{(n)} \Rightarrow \gamma^*\)}
  \end{prooftree}
  where I.H.\ indicates an application of the induction hypothesis to \(\gamma\) and \(*\mathrm{R}_1\) is applied \(n\) times.
  Then, the desired proof is
  \[
    \AxiomC{}
    \RightLabel{\(*\mathrm{R}_0\)}
    \UnaryInfC{\(\Rightarrow \gamma^*\)}
    \AxiomC{\(\cdots\)}
    \AxiomC{\(\tau_n\)}
    \noLine
    \UnaryInfC{\(\gamma^{(n)}\Rightarrow \gamma^*\)}
    \AxiomC{\(\cdots\)}
    \QuaternaryInfC{\(\gamma^* \Rightarrow \gamma^*\)}
    \DisplayProof
  \]

  To show that \(0\mathrm{R}\) is admissible we proceed by an induction on the height of proofs.
  So assume that \(\pi \vdash \Gamma \Rightarrow 0\) in \(\gentzenContinuousActionLatticeOmega + \mathcal{R}\), we do cases in the last step of \(\pi\).
  We notice that it is impossible for the last rule of \(\pi\) to be a principal right rule different from \(0\mathrm{R}\) or to be id, since \(0\) is the formula in the right hand side of the conclusion of \(\pi\).
  We will cover 3 cases, leaving the rest for the reader.

  Case \(0 \mathrm{L}\).
  In this case \(\pi\) is of shape:
  \[
    \AxiomC{}
    \RightLabel{\(0\mathrm{L}\)}
    \UnaryInfC{\(\Gamma, 0, \Delta \Rightarrow 0\)}
    \DisplayProof
  \]
  and then the desired proof is
  \[
    \AxiomC{}
    \RightLabel{\(0\mathrm{L}\)}
    \UnaryInfC{\(\Sigma_l,\Gamma, 0, \Delta,\Sigma_r \Rightarrow \beta\)}
    \DisplayProof
  \]

  Case \(\rresidual \mathrm{L}\).
  In this case \(\pi\) is of shape:
  \[
    \AxiomC{\(\pi_0\)}
    \noLine
    \UnaryInfC{\(\Delta \Rightarrow \alpha_0\)}
    \AxiomC{\(\pi_1\)}
    \noLine
    \UnaryInfC{\(\Gamma, \alpha_1, \Sigma \Rightarrow 0\)}
		\RightLabel{\(\rresidual\)L}
    \BinaryInfC{\(\Gamma, \alpha_1 \rresidual \alpha_0, \Delta, \Sigma \Rightarrow 0\)}
		\DisplayProof
  \]
  Since \(\height(\pi_1) < \height(\pi)\) we can apply the induction hypothesis and obtain a \(\tau \vdash \Sigma_l, \Gamma, \alpha_1, \Sigma, \Sigma_r \Rightarrow \beta\).
  Then, the desired proof is
  \[
    \AxiomC{\(\pi_0\)}
    \noLine
    \UnaryInfC{\(\Delta \Rightarrow \alpha_0\)}
    \AxiomC{\(\tau\)}
    \noLine
    \UnaryInfC{\(\Sigma_l,\Gamma, \alpha_1, \Sigma, \Sigma_r \Rightarrow \beta\)}
		\RightLabel{\(\rresidual\)L}
    \BinaryInfC{\(\Sigma_l,\Gamma, \alpha_1 \rresidual \alpha_0, \Delta, \Sigma, \Sigma_r \Rightarrow \beta\)}
		\DisplayProof
  \]

  Case of an analytical rule \(R\).
  Let \(R\) be
  \[
    \AxiomC{\(\Gamma', \Upsilon_0, \Delta' \Rightarrow \alpha\)}
    \AxiomC{\(\cdots\)}
    \AxiomC{\(\Gamma', \Upsilon_{n-1}, \Delta' \Rightarrow \alpha\)}
    \RightLabel{\(R\)}
    \TrinaryInfC{\(\Gamma', \Upsilon, \Delta' \Rightarrow \alpha\)}
    \DisplayProof
  \]
  where \(\Gamma', \Delta'\) are sequence metavariables, \(\Upsilon_0\), \ldots, \(\Upsilon_{n-1}\), \(\Upsilon\) are sequences of sequence metavariables and \(\alpha\) is a formula metavariable.
  In this case \(\pi\) is of shape:
  \[
    \AxiomC{\(\pi_0\)}
    \noLine
    \UnaryInfC{\(\Gamma, \Sigma_0, \Delta \Rightarrow 0\)}
    \AxiomC{\(\cdots\)}
    \AxiomC{\(\pi_{n-1}\)}
    \noLine
    \UnaryInfC{\(\Gamma, \Sigma_{n-1}, \Delta \Rightarrow 0\)}
    \RightLabel{\(R\)}
    \TrinaryInfC{\(\Gamma, \Sigma, \Delta \Rightarrow 0\)}
    \DisplayProof
  \]
  where \(\Gamma\) is obtained instantiating \(\Gamma'\), \(\Delta\) is obtained instantiating \(\Delta'\), \(\Sigma_i\) is obtained instanting the sequence metavariables of \(\Upsilon_i\), \(\Sigma\) is obtained by instantiating the sequence metavariables of \(\Upsilon\) and \(0\) is obtained instantiating \(\alpha\).
  Since \(\height(\pi_i) < \height(\pi)\) for any \(i < n\) we can apply the induction hypothesis to obtain \(\tau_i \vdash \Sigma_l, \Gamma, \Sigma_i, \Delta, \Sigma_r \Rightarrow \beta\).
  Then the desired proof is
  \[
    \AxiomC{\(\tau_0\)}
    \noLine
    \UnaryInfC{\(\Sigma_l,\Gamma, \Sigma_0, \Delta,\Sigma_r \Rightarrow \beta\)}
    \AxiomC{\(\cdots\)}
    \AxiomC{\(\tau_{n-1}\)}
    \noLine
    \UnaryInfC{\(\Sigma_l, \Gamma, \Sigma_{n-1}, \Delta, \Sigma_r \Rightarrow \beta\)}
    \RightLabel{\(R\)}
    \TrinaryInfC{\(\Sigma_l,\Gamma, \Sigma, \Delta, \Sigma_r\Rightarrow \beta\)}
    \DisplayProof
  \]
  where in the application of \(R\) we are instantiating \(\Gamma'\) with \(\Sigma_l, \Gamma\), \(\Delta'\) with \(\Delta, \Sigma_r\), \(\alpha\) with \(\beta\) and the rest of metavariables are instantiated as before.
\end{proof}

An \emph{analytic} quasiequation is a quasiequation of the form 
\[
\alpha_1 \leq \beta \mathbin{\&} \dots \mathbin{\&} \alpha_{n} \leq \beta \implies \alpha_{0} \leq \beta,
\]
where $\alpha_0 = x_1\cdots x_m$ for distinct variables $x_1,\dots, x_m$, $\beta =y$ for some variable $y$ distinct from $x_1,\dots, x_m$, and $\alpha_i$ is a product of variables from $x_1,\dots, x_m$ for $i>0$, noting that the empty product is stipulated to be equal to $1$. Note that  a residuated lattice satisfies such an analytic quasiequation if and only if it satisfies the equation $\alpha_0 \leq \alpha_1 \vee \dots \vee \alpha_n$. See $\cite{CGT2012}$ for more details.

We fix a bijection between the denumerable set of metavariables (containing both the formula and sequence metavariables) and we define for every sequence of metavariables $\Upsilon$ a term $t(\Upsilon)$ by replacing each metavariable by their corresponding variable, replacing commas with $\cdot$, and stipulating $t(\Upsilon) = 1$ if $\Upsilon$ is empty. Accordingly we define for every structural rule
 \[ 
    \AxiomC{\(\Upsilon_0 \Rightarrow \alpha_0\)}
    \AxiomC{\(\ldots\)}
    \AxiomC{\(\Upsilon_{k-1} \Rightarrow \alpha_{k-1}\)}
    \RightLabel{\(R\)}
    \TrinaryInfC{\(\Upsilon_{k} \Rightarrow \alpha_{k}\)}
    \DisplayProof
 \]
a quasiequation, denoted $q(R)$, by
\[
t(\Upsilon_0) \leq t(\alpha_0) \mathbin{\&} \dots \mathbf{\&} t(\Upsilon_{k-1}) \leq t(\alpha_{k-1}) \implies  t(\Upsilon_k) \leq t(\alpha_k).
\]
For example, 
\[
q(\text{Cut}) = (x \leq y \mathbin{\&} zyw \leq u) \implies zxw\leq u.
\]
For a set of structural rules $\mathcal{R}$ we define $Q(\mathcal{R}) = \{q(R) \mid R\in \mathcal{R} \}$.

In the special case of an analytic structural rule
  \[
    \AxiomC{\(\Gamma, \Upsilon_0, \Delta \Rightarrow \beta\)}
    \AxiomC{\(\cdots\)}
    \AxiomC{\(\Gamma, \Upsilon_{n}, \Delta \Rightarrow \beta\)}
    \RightLabel{\(R\)}
    \TrinaryInfC{\(\Gamma, \Upsilon, \Delta \Rightarrow \beta\)}
    \DisplayProof
  \]
we define
\[ 
q_a(R) := t(\Upsilon_0) \leq t(\beta) \mathbin{\&} \dots \mathbin{\&} t(\Upsilon_{n}) \leq t(\beta) \implies t(\Upsilon) \leq t(\beta),
\]
i.e., we forget about the context. Note that if $R$ is an analytic rule then $q_a(R)$ is an analytic quasiequation. For a set of analytic structural rules $\mathcal{R}$ we define $Q_a(\mathcal{R}) = \{q_A(R) \mid R\in \mathcal{R} \}$.

Given a set \(\mathcal{R}\) of linear rules, one can show soundness of \(\gentzenContinuousActionLatticeOmega + \mathcal{R}\) with respect to the class of \(*\)-continuous action lattices satisfying the quasiequations \(Q(\mathcal{R})\) simply by an induction in the ordinal height of proofs. In particular, if every rule in $\mathcal{R}$ is analytic, then \(\gentzenContinuousActionLatticeOmega + \mathcal{R}\) is sound with respect to the class of \(*\)-continuous action lattices satisfying the quasiequations \(Q_a(\mathcal{R})\).

In order to define the non-wellfounded system we need the concept of a progressing thread.

\begin{definition}
	Let \(\pi\) be a preproof in some proof system and \(b \in \branches(\pi)\).
	A \emph{thread in \((\pi, b)\)} is a function
	\(t \colon I \longrightarrow \NN \cup \set{-1}\) such
	that
	\begin{enumerate}
		\item \(I\) is a non-empty interval in \(\mathbb{N}\).
		\item If \(i \in I\) then \(t_i\) is a formula ocurrence in
			\(S^\pi_{b_i}\), i.e.\ the sequent in \(\pi\) at node
			\(b_i\).
		\item If \(i, i + 1 \in I\) then \(t_{i+1}\) is an immediate ancestor of
			\(t_i\).
	\end{enumerate}

	A thread \(t \colon I \longrightarrow \NN \cup \set{-1}\) is said to be a
	\(*\)-thread if there is a formula \(\alpha\) such that for any \(i \in I\)
	we have that \(S^\pi_{b_i}(t_i) = \alpha^*\)and \(t_i \neq -1\),
	i.e.\ it always denotes the same \(*\)-formula and stays on the left hand side
	of the sequent.
	Given threads \(t, t'\) we say that \(t'\) is a suffix of \(t\) if 
	\begin{enumerate}
		\item \(\domain(t')\) is a suffix of \(\domain(t)\).
		\item For any \(i \in \domain(t) \cap \domain(t')\) we have \(t_i = t'_i\).
	\end{enumerate}
	
	We call a thread \(t\) in \((\pi,b)\) \emph{progressing} if
	\begin{enumerate}
		\item \(t\) has a suffix \(t'\) which is a \(*\)-thread, and
    \item there are infinitely many $i\in \domain(t')$ such that $t'_i$ is principal in \(\pi\) at node \(b_i\).
	\end{enumerate}
	We say that \((\pi, b)\) is \emph{progressing} if it has a progressing
	thread.
	By definition any progressing thread gives a progressing 
	\(*\)-thread.
\end{definition}

\begin{definition}
Let \(\mathcal{R}\) be a set of linear rules.
We define \(\gentzenContinuousActionLattice + \mathcal{R}\) as the proof system
whose rules are \(\text{id}\), those in \(\mathcal{R}\) and the displayed rules in Figure~\ref{fig:common-rules} without \(\text{*}L_\omega\).
The branch condition is that any branch has a progressing thread.
\end{definition}

In \(\gentzenContinuousActionLattice + \mathcal{R}\), the only point where a \(*\)-thread can be principal is at the rule \(*\text{L}\),
since it goes through the left hand side of sequents only.

Note that in non-wellfounded proofs there is, a priori, no notion of ordinal height, so we cannot do induction in the height of proofs.
For this reason, to show soundness of \(\gentzenContinuousActionLattice + \mathcal{R}\) with respect to its algebraic semantics, we will use a translation from \(\gentzenContinuousActionLattice + \mathcal{R}\) to \(\gentzenContinuousActionLatticeOmega + \mathcal{R}\).
In \cite{DP2018}, there is a simpler soundness proof for \(\gentzenContinuousActionLattice\) (without any structural rules).
This proof is tantamount to translating \(\gentzenContinuousActionLattice\) to \(\gentzenContinuousActionLatticeOmega\). The latter translation is recursive, but the translation we will define in Subsection~\ref{subsec:translation} is corecursive.

We note that an analogue of Lemma~\ref{l:adm-of-frame-rules} can be proven for \(\gentzenContinuousActionLatticeOmega + \mathcal{R}\) where \(\mathcal{R}\) is a set of analytical rules.
The proof that \(\alpha \Rightarrow \alpha\) is provable will be again by induction in the size of \(\alpha\), while the proof that \(0\mathrm{R}\) is admissible needs to be proven using corecursion.
This, together with other facts like the derivability of the \(*\mathrm{L}_\omega\) in \(\gentzenContinuousActionLatticeOmega\) proven below, allows to show (cut-free) completeness of \(\gentzenContinuousActionLattice + \mathcal{R}\) directly.
However, since in order to prove soundness for \(\gentzenContinuousActionLattice + \mathcal{R}\) we need to transform its proofs into proofs of \(\gentzenContinuousActionLatticeOmega + \mathcal{R}\), as explained in the previous paragraph, we have opted to show (cut-free) completeness via this transformation.
So a proof of soundness for \(\gentzenContinuousActionLattice + \mathcal{R}\) without stepping into wellfounded proofs will make the procedures of algebraic cut-elimination for \(\gentzenContinuousActionLattice + \mathcal{R}\) and \(\gentzenContinuousActionLatticeOmega + \mathcal{R}\) completely independant.

\begin{lemma}\label{lm:wellfounded-to-nonwellfounded}
  The rule $*\mathrm{L}_\omega$ is derivable in $\gentzenContinuousActionLattice$. Consequently, if \(\mathcal{R}\) is a any set of rules and \(\gentzenContinuousActionLatticeOmega + \mathcal{R} \vdash S\), then \(\gentzenContinuousActionLattice + \mathcal{R} \vdash S\).
\end{lemma}
\begin{proof}
For any natural number \(i\) let as define the sequent
\[
  S_i := \Gamma, \alpha^{(i)}, \Delta \Rightarrow \beta.
\]
Then, we have the following derivation in \(\gentzenContinuousActionLattice\)
\[
  \AxiomC{\(\Gamma, \Delta \Rightarrow \beta\)}
  \AxiomC{\(\Gamma, \alpha, \Delta \Rightarrow \beta\)}
  \AxiomC{\(\Gamma, \alpha^{(n)}, \Delta \Rightarrow \beta\)}
  \AxiomC{\(\vdots\)}
  \noLine
  \UnaryInfC{\(S_{n+1}\)}
  \RightLabel{\(*\mathrm{L}\)}
  \BinaryInfC{\(S_n\)}
  \noLine
  \UnaryInfC{\(\vdots\)}
  \UnaryInfC{\(S_2\)}
  \RightLabel{\(*\mathrm{L}\)}
  \BinaryInfC{\(S_1\)}
  \BinaryInfC{\(\Gamma, \alpha^*, \Delta \Rightarrow \beta\)}
  \DisplayProof
\]
This derivation clearly fulfills the branch condition in the unique branch it has, since we can follow the unique \(*\)-thread of the displayed \(\alpha^*\) at the conclusion obtaining a progressing thread.

Then, to show that any proof \(\pi\) in \(\gentzenContinuousActionLatticeOmega + \mathcal{R}\) can be converted into a proof in \(\gentzenContinuousActionLattice + \mathbb{R}\) we simply do a recursion on \(\height(\pi)\), converting any instance of \(*\mathrm{L}_\omega\) into the derivation defined above.
\end{proof}

\subsection{From non-wellfounded to wellfounded proofs}
\label{subsec:translation}
In this subsection we fix a set of linear rules \(\mathcal{R}\).

Using methods similar to those of \cite{das_et_al:LIPIcs.FSTTCS.2023.40}, we show how to transform proofs in \(\gentzenContinuousActionLattice + \mathcal{R}\) into proofs in \(\gentzenContinuousActionLatticeOmega + \mathcal{R}\).
This together with Lemma~\ref{lm:wellfounded-to-nonwellfounded} implies that \(\gentzenContinuousActionLatticeOmega + \mathcal{R}\) and \(\gentzenContinuousActionLattice + \mathcal{R}\) prove the same sequents.
This will allow us to show that soundness with respect to algebraic semantics.

In this section, we work with some slightly altered proof systems.
In particular, we extend both proof systems with a rule \(\cdot \text{L} + 1\) defined as
\[
  \AxiomC{\(\Pi_l, \fbox{\(\alpha_0\)}, \fbox{\(\alpha_1\)} \Rightarrow \beta\)} 
  \RightLabel{\(\cdot \text{L} + 1\)}
  \UnaryInfC{\(\Pi_l, \fbox{\(\alpha_0 \cdot \alpha_1\)} \Rightarrow \beta\)}
  \DisplayProof
\]
This rule looks exactly equal to \(\cdot \text{L}\), but we impose that if the node \(w\) is the conclusion of this rule, then its unique child is \(w1\).
In words, even if the rule is unary the unique premise ``goes to the right''.
Also, we change the \(\omega\)-rule in the wellfounded system \(\gentzenContinuousActionLatticeOmega + \mathcal{R}\) for
\[
  \AxiomC{\(\Pi_l, \Pi_r \Rightarrow \beta\)}
  \AxiomC{\((\Pi_l, \fbox{\(\alpha\)}, \fbox{\(\alpha^i\)}, \Pi_r \Rightarrow \beta)_{i \in \NN}\)}
  \RightLabel{\(*\text{L}_\omega\)}
  \BinaryInfC{\(\Pi_l, \fbox{\(\alpha^*\)}, \Pi_r \Rightarrow \beta\)}
  \DisplayProof
\]
We denote this modified system as \({\gentzenContinuousActionLatticeOmega}' + \mathcal{R}\).

Then to translate a proof \(\pi\) in \(\gentzenContinuousActionLattice + \mathcal{R}\) to \(\gentzenContinuousActionLatticeOmega + \mathcal{R}\), we start by noticing that \(\pi\) is already a proof in \(\gentzenContinuousActionLattice + (\cdot \text{L} + 1) + \mathcal{R}\).
So we can apply the method defined in this subsection, obtaining a proof \(\tau\) in \({\gentzenContinuousActionLatticeOmega}' + (\cdot \text{L} + 1) + \mathcal{R}\).
It is clear how to convert this proof \(\tau\) into a proof \(\tau'\) in \({\gentzenContinuousActionLatticeOmega}' + \mathcal{R}\).

Finally, by induction in the ordinal height, we can convert \(\tau'\) into a proof \(\rho\) in \(\gentzenContinuousActionLatticeOmega + \mathcal{R}\), using Lemma~\ref{lm:invertible} below; this lemma is itself proven by induction in the height of \(\pi\).

\begin{lemma}\label{lm:invertible}
  The rule \(\cdot \text{L}\) is height-preserving invertible in \(\gentzenContinuousActionLatticeOmega + \mathcal{R}\), i.e., if \(\pi \vdash \Gamma, \alpha \cdot \beta, \Delta \Rightarrow \gamma\) in \(\gentzenContinuousActionLatticeOmega + \mathcal{R}\), then there is a proof \(\tau \vdash \Gamma, \alpha, \beta,\Delta\Rightarrow \gamma\) such that \(\height(\tau) \leq \height(\pi)\).
\end{lemma}

We will write \(\nwProof, \wProof\) to denote the sets of proofs of \(\gentzenContinuousActionLattice + (\cdot \text{L} + 1) + \mathcal{R}\) and \({\gentzenContinuousActionLatticeOmega}' + (\cdot \text{L} +1) + \mathcal{R}\).
Similarly, \(\text{p}\nwProof\) and \(\text{p}\wProof\) will be used to denote the corresponding sets of preproofs.

Given a sequent \(S =(\Gamma \Rightarrow \beta)\) we say that \(S' = (\Gamma' \Rightarrow \beta)\) is a \(*\)-projection of \(S\) if \(\Gamma'\) is obtained by replacing the \(*\) of some \(*\)-formulas in \(\Gamma\) by some natural number (not necessarily the same in each formula).
The desired translation will be constructed using that if \(\gentzenContinuousActionLattice + (\cdot \text{L} + 1) + \mathcal{R}\) proves \(S\) it also proves \(S'\).

\begin{definition}
  Let \(S = \alpha_0,\ldots, \alpha_{n-1} \Rightarrow \beta\) be a sequent and further let 
  \(f \colon \set{0,\ldots, n - 1} \longrightarrow \NN_\bot\).
  We say that \(f\) is a \emph{\(*\)-assignment for} \(S\) if
  \(f(i) \neq \bot\) implies that \(\alpha_i\) is a \(*\)-formula.
  
  If \(f\) is a \(*\)-assignment of \(S\), we define the
  \emph{\(f\)-projection of \(S\)}, denoted as \(S^f\),
  to be the sequent \(\alpha'_0, \ldots, \alpha'_n \Rightarrow \beta\) where
  \begin{enumerate}
    \item \(\alpha'_i := \alpha_i\) if \(f(i) = \bot\).
    \item \(\alpha'_i := \beta^n_i\) if \(f(i) = n\) and \(\alpha_i = \beta_i^*\). \qedhere
  \end{enumerate}
\end{definition}

In case we denote \(S\) via its components (either with formulas or formula sequences) then we will denote \(S^f\) by adding a \(^f\) to each of its components in the LHS.
For example, if \(S = \Gamma, \alpha_0, \alpha_1,\Delta \Rightarrow \beta\), then \(\Gamma^f, \alpha^f_0, \alpha^f_1, \Delta^f \Rightarrow \beta\) denotes \(S^f\).

Given \(f\) a \(*\)-assignment for \(S\) and a rule instance
  \[
    \AxiomC{\(S_0\)}
    \AxiomC{\(\cdots\)}
    \AxiomC{\(S_{n-1}\)}
    \TrinaryInfC{\(S\)}
    \DisplayProof
  \]
  which is either linear or principal with \(f(\phi) = \bot\) for the principal formula \(\phi\),
  we define
  \label{df:evolution-of-assignment}
  \[
    f_i(q) =
    \begin{cases}
    f(q') &\text{if \(q\) is an immediate ancestor} \\
            &\text{of some }q'\text{ in }S_i, \\
    \bot &\text{if there is no \(q'\) in \(S_i\) such that} \\ 
         &\text{\( q\) is an immediate ancestor of \(q'\)}.
    \end{cases}
  \]
  This is a \(*\)-assignment function for \(S\).
  
  The next lemma is an easy consequence of the definition of linear rules.

  \begin{lemma}\label{lm:rule-uniformity}
  Let \(f\) be a \(*\)-assignment for \(S\) and 
  \[
    \AxiomC{\(S_0\)}
    \AxiomC{\(\cdots\)}
    \AxiomC{\(S_{n-1}\)}
    \TrinaryInfC{\(S\)}
    \DisplayProof
  \]
  be an instance of a rule \(R\) which is either linear, \(\text{id}\) or principal with principal formula \(\phi\) such that \(f(\phi) = \bot\).
  Then
  \[
    \AxiomC{\(S_0^{f_0}\)}
    \AxiomC{\(\cdots\)}
    \AxiomC{\(S_{n-1}^{f_{n-1}}\)}
    \TrinaryInfC{\(S^f\)}
    \DisplayProof
  \]
  is also an instance of \(R\).
\end{lemma}
\begin{proof}
  In structural rules the result follows from the restriction of all the metavariables in the conclusion being pairwise distinct.
  The case for the identity rule is trivial, since the only formulas appearing are atomic so the \(*\)-assignment will not change the sequent.
  For principal rules we just notice that since, the principal formula is unchanged by \(f\), so are the auxiliaries. 
  Then only the context changes.
  However, it is easy to see that it changes uniformly, i.e.\ each of the instantiations of the metavariables for sequences are replaced by a new instantiation, giving rise to another instance of the rule.
\end{proof}

\begin{definition}
  Let \(n \in \NN\), we define the projection function \(\Projection\) from
  \[
    \set{(\pi, f) \mid \pi \in \text{p}\mathbb{P}^\infty,
    f\, *\text{-assignment for }S^\pi}
  \]
  to \(\text{p}\mathbb{P}^\infty\).
  We define it corecursively, writing \(\Projection(\pi,f)\) as \(\Projection^f(\pi)\).

  If \(k\) is principal and \(f(k) = 0\), we define
  \[
    \hspace{-0.4cm}
    \Projection^f\left(
      \AxiomC{\(\pi_0\)}
      \noLine
      \UnaryInfC{\(\Pi_l, \Pi_r \Rightarrow \beta\)}
      \AxiomC{\(\pi_1\)}
      \noLine
      \UnaryInfC{\(\Pi_l, \alpha,\alpha^*, \Pi_r \Rightarrow \beta\)}
      \RightLabel{\(*\)L}
      \BinaryInfC{\(\Pi_l, \alpha^*, \Pi_r \Rightarrow \beta\)}
      \DisplayProof
    \right) {=}
  \]

  \[
    \AxiomC{\(\Projection^{f'}(\pi_0)\)}
    \noLine
    \UnaryInfC{\(\Pi^{f'}_l, \Pi^{f'}_r \Rightarrow \beta\)}
    \RightLabel{\(1\text{L}\)}
    \UnaryInfC{\(\Pi^f_l,\alpha^{0}, \Pi^f_r \Rightarrow \beta\)}
    \DisplayProof
  \]
  where \(\alpha^*\) in the conclusion is at position \(k\) and
  \[
    f' := (f \restricts k) \cup \set{(k'-1, n) \mid 
      (k', n) \in f, k' > k}.
  \]

  If \(k\) is principal and \(f(k) \in \mathbb{N}\setminus \set{0}\), we define
  \[
    \hspace{-0.4cm}
    \Projection^f\left(
      \AxiomC{\(\pi_0\)}
      \noLine
      \UnaryInfC{\(\Pi_l, \Pi_r \Rightarrow \beta\)}
      \AxiomC{\(\pi_1\)}
      \noLine
      \UnaryInfC{\(\Pi_l, \alpha, \alpha^*, \Pi_r \Rightarrow \beta\)}
      \RightLabel{\(*\)L}
      \BinaryInfC{\(\Pi_l, \alpha^*, \Pi_r \Rightarrow \beta\)}
      \DisplayProof
    \right) =
  \]
  \[
    \AxiomC{\(\Projection^{f'}(\pi_1)\)}
    \noLine
    \UnaryInfC{\(\Pi^{f'}_l, \alpha, \alpha^{n-1}, \Pi^{f'}_r \Rightarrow \beta\)}
    \RightLabel{\(\cdot\text{L} + 1\)}
    \UnaryInfC{\(\Pi^{f'}_l, \alpha^{n}, \Pi^{f'}_r \Rightarrow \beta\)}
    \DisplayProof
  \]
  where \(\alpha^*\) at the conclusion is at position \(k\) and
  \begin{multline*}
    f' := (f\restricts k) \cup \set{(k,\bot),(k+1, n-1)} \cup \\
      \set{(k'+1, n) \mid (k', n) \in f \text{ and }k' > k}.
    \end{multline*}
  
  If \(k\) is principal and \(f(k) = \bot\) or there is no principal formula, we define
  \[
    \Projection^f\left(
    \AxiomC{\(\ldots\)}
    \AxiomC{\(\pi_{i}\)}
    \noLine
    \UnaryInfC{\(S_i\)}
    \AxiomC{\(\ldots\)}
    \RightLabel{\(r\)}
    \TrinaryInfC{\(S\)}
    \DisplayProof
    \right) =
    \AxiomC{\(\ldots\)}
    \AxiomC{\(\Projection^{f_i}(\pi_i)\)}
    \noLine
    \UnaryInfC{\(S^{f_i}_i\)}
    \AxiomC{\(\ldots\)}
    \RightLabel{\(r\)}
    \TrinaryInfC{\(S^f\)}
    \DisplayProof
  \]
  where the \(f_i\)'s were defined above.
  Finally, we note that it is easy to check that the result lies in \(\text{p}\mathbb{P}^\infty\) in the first two cases.
  In the third case we just need to use Lemma~\ref{lm:rule-uniformity}.
\end{definition}

The following is proven by induction on the length of nodes.
\begin{lemma}\label{lm:nodes-of-projection-are-nodes-of-original}
  For  \(\text{p}\mathbb{P}^\infty\), \(\nodes(\Projection^f(\pi)) \subseteq \nodes(\pi)\).
  In particular,
  \(\branches(\Projection^f(\pi)) \subseteq \branches(\pi)\).
\end{lemma}

If we project the starting point of a progressing thread in a branch \(b\), then the branch \(b\) cannot also be a branch of the projection, as the following lemma records.
\begin{lemma}\label{lm:star-progressing-thread-implies-stucked}
  Let \(\pi \in \mathrm{p}\nwProof\),
  \(b \in \branches(\pi)\),
  and \(f\) be a \(*\)-assignment for \(S^\pi\).
  If there is a \((\pi,b)\)-progressing \(*\)-thread \(t\) starting at \(0\) such
  that \(f(t_0) \neq \bot\), then \(b \not \in \branches(\Projection^f(\pi))\).
\end{lemma}
\begin{proof}
  We do induction in the pairs \((f(t_0), m)\) ordered lexicographically, where
  \(m\) is the distance to the next time the thread \(t\) is principal in
  \((\pi,b)\).

  Case 1: \(t\) is principal in the last step of \(\pi\) and \(f(t_0) = 0\).
  Then \(\pi\) and \(\Projection^f(\pi)\) are, respectively, of shape
    \[
      \AxiomC{\(\pi_0\)}
      \noLine
      \UnaryInfC{\(\Pi_l, \Pi_r \Rightarrow \beta\)}
      \AxiomC{\(\pi_1\)}
      \noLine
      \UnaryInfC{\(\Pi_l, \alpha, \alpha^* \Pi_r \Rightarrow \beta\)}
      \BinaryInfC{\(\Pi_l, \alpha^*, \Pi_r \Rightarrow \beta\)}
      \DisplayProof
      \hspace{1cm}
      \AxiomC{\(\Projection^{f'}(\pi_0)\)}
      \noLine
      \UnaryInfC{\(\Pi'_l, \Pi'_r \Rightarrow \beta\)}
      \RightLabel{\(1\)L}
      \UnaryInfC{\(\Pi'_l, 1, \Pi'_r \Rightarrow \beta\)}
      \DisplayProof
    \]
  where the displayed \(\alpha^*\) in the conclusion is at position \(t_0\).
  However, since \(t\) is a \(*\)-thread which progresses infinitely often it
  must be the case that \(b = 1 : \tail(b)\), i.e.\ that it goes to the right.
  We can see that \(1\) is not a node in \(\Projection^f(\pi)\), so
  \(b \not \in \branches(\Projection^f(\pi))\).

  Case 2: \(t\) is principal in the last step of \(\pi\) and \(f(t_0) = n > 0\).
  Then \(\pi\) and \(\Projection^f(\pi)\) are, respectively, of shape
  \[
      \AxiomC{\(\pi_0\)}
      \noLine
      \UnaryInfC{\(\Pi_l, \Pi_r \Rightarrow \beta\)}
      \AxiomC{\(\pi_1\)}
      \noLine
      \UnaryInfC{\(\Pi_l, \alpha, \alpha^* \Pi_r \Rightarrow \beta\)}
      \RightLabel{\(*\text{L}\)}
      \BinaryInfC{\(\Pi_l, \alpha^*, \Pi_r \Rightarrow \beta\)}
      \DisplayProof
    \]
    \[
      \AxiomC{\(\Projection^{f'}(\pi_1)\)}
      \noLine
      \UnaryInfC{\(\Pi'_l, \alpha, \alpha^{n-1}, \Pi'_r \Rightarrow \beta\)}
      \RightLabel{\(\cdot \text{L} + 1\)}
      \UnaryInfC{\(\Pi'_l, \alpha^n, \Pi'_r \Rightarrow \beta\)}
      \DisplayProof
    \]
  where \(\alpha^*\) in the conclusion is at position \(t_0\).
  Again, as \(t\) is a \(*\)-thread which progresses infinitely often it
  must be the case that \(b = 1 : \tail(b)\), i.e.\ that it goes to the right premise.
  We can apply the induction hypothesis to \(\pi_1 \in \text{p}\nwProof\),
  \(\tail(b) \in \branches(\pi_1)\), \(f'\) \(*\)-assignment for \(S^{\pi_1}\)
  and \(\tail(t)\), since it is a \((\pi_1, \tail(b))\)-progressing \(*\)-thread
  starting at \(0\) such that \(f'(\tail(t)_0) = f'(t_1) = f(t_0) - 1 \neq \bot\),
  as the measure has decreased.

  Case 3: \(t\) is not principal in the last step of \(\pi\).
  In this case either \(t_0\) belongs to the context of a principal rule or the
  last rule is structural.
  Also \(b = i : \tail(b)\), for some \(i\).
  \(\pi\) and \(\Projection^f(\pi)\) will be of shape:
    \[ 
      \AxiomC{\(\cdots\)}
      \AxiomC{\(\pi_i\)}
      \noLine
      \UnaryInfC{\(S_i\)}
      \AxiomC{\(\cdots\)}
      \RightLabel{\(r\)}
      \TrinaryInfC{\(S\)}
      \DisplayProof
      \hspace{1cm}
      \AxiomC{\(\cdots\)}
      \AxiomC{\(\Projection^{f_i}(\pi_i)\)}
      \noLine
      \UnaryInfC{\(S^{f_i}_i\)}
      \AxiomC{\(\cdots\)}
      \RightLabel{\(r'\)}
      \TrinaryInfC{\(S'\)}
      \DisplayProof
    \]
  Then will have that \(\pi_i \in \text{p}\nwProof\),
  \(\tail(b) \in \branches(\pi_i)\),
  \(f_i\) is a \(*\)-assignment for \(S^{\pi_i}\),
  and \(\tail(t)\) will be a \((\pi_1, \tail(b))\)-progressing \(*\)-thread
  starting at \(0\) such that \(f_i(\tail(t)_0) = f_i(t_1) = f(t_0) \neq \bot\),
  since \(t_1\) must be an immediate ancestor of \(t_0\).
  However, notice that the \(i\) such that \(\tail(t)_i\) is principal must be 
  smaller than the first \(i\) such that \(t_i\) is principal (as \(t_0\) is not
  principal).
  This mean that in the order the first number remains the same
  (\(f_i(t_1) = f(t_0)\)) while the second number strictly decreases.
  We can finish the proof just by applying the induction hypothesis.
\end{proof}

The following corollary intuitively says that progressing \(*\)-threads and projected \(*\)-threads cannot coincide at any point.

\begin{corollary}\label{cor:progressing-distinct-from-all-principal}
  Let \(\pi \in \nwProof\), 
  \(f\) be a \(*\)-assignment for \(S^\pi\)
  and \(b \in \branches(\Projection^f(\pi))\).
  If \(t\) is a \((\pi,b)\)-progressing \(*\)-thread, then for each  \((\pi,b)\) \(*\)-thread $t'$ starting at \(0\) with \(f(t'_0) \neq \bot\), we have \(t_i \neq t'_i\) for any \(i \in \domain(t) \cap \domain(t')\).
\end{corollary}
\begin{proof}
  Assume \(t_i = t'_i\), we can define a thread which is equal to \(t'\) until step \(i\) and from there is equal to \(t\).
  This will be a \((\pi,b)\)-progressing (as \(t\) is a suffix of it) thread which is projected by \(f\), so \(b \not \in \branches(\Projection^f(\pi))\) by the previous Lemma, a contradiction.
\end{proof}

\begin{definition}
  Let \(\pi \in \text{p}\nwProof\) and \(b \in \branches(\pi)\).
  Given a thread \(t\) in \((\pi,b)\) we will say that it is \emph{principal at step \(i\)} if
  \begin{enumerate}
    \item \(i \in \domain(t)\),
    \item the rule at node \(b_i\) of \(\pi\) is a principal rule, and
    \item \(t(i)\) is the principal formula (ocurrence) at node \(b_i\) of \(\pi\).
  \end{enumerate}
  Then, we say that \(i\) is a \emph{progress point} of \((\pi, b)\) if
  there is a \((\pi,b)\)-progressing \(*\)-thread \(t\) such that \(t\) is principal at step \(i\) in \((\pi,b)\).
  We will write \(\progressPoints(\pi, b)\) to denote the set of progress points of \((\pi, b)\).

  Notice that \((\pi, b)\) fulfills the branch condition iff \(\progressPoints(\pi, b) \neq \varnothing\).
  
  We define the \emph{critical height} of \((\pi, b)\) as
  \[ 
    \criticalHeight(\pi, b) = \begin{cases}
      \min \progressPoints(\pi, b) &\text{if }\progressPoints(\pi, b) \neq
      \varnothing,\\
      \infty &\text{otherwise}.
    \end{cases}
  \]
  With the previous observation we obtain that any progressing branch (and any subbranch of it) has a finite critical height.
  Also,
  if \(0 < \criticalHeight(\pi, b) < \infty\) then \(\criticalHeight(\pi_{b_0}, \tail(b)) < \criticalHeight(\pi, b)\),
  since
  \[
    \progressPoints(\pi_{b_0}, \tail(b)) =
    \set{i \in \NN \mid i + 1 \in \progressPoints(\pi, b)}.
    \qedhere
  \]
\end{definition}

\begin{lemma}\label{lm:projection-and-threads}
  Let \(\pi \in \text{p}\nwProof\), 
  \(f\) be a \(*\)-assignment for \(S^\pi\),
  \(b \in \branches(\Projection^f(\pi))\),
  and \(t\) be a thread of \((\pi, b)\).
  Let \({(S_i, r_i)}_{i \in \NN}\) be the sequents and rules in branch \(b\) of
  \(\pi\) and \({(S'_i, r'_i)}_{i \in \NN}\) be the sequents and rules in branch
  \(b\) of \(\Projection^f(\pi)\).
  Then
  \begin{enumerate}
    \item \(t\) is also a thread of  \((\Projection^f(\pi), b)\).
    \item If \(t\) is principal at step \(i\) in \((\pi, b)\),
      then \(t\) is principal at step \(i\) in \((\Projection^f(\pi),b)\).
    \item If \(i \in \domain(t)\) and for all \(*\)-thread \(t'\) in \((\pi,b)\) starting at \(0\) with \(f(t'_0) \neq \bot\)
    such that \(i \in \domain(t')\) we have \(t_i \neq t'_i\),
    then \(S'_i(t_i) = S_i(t_i)\).
  \end{enumerate}
\end{lemma}
\begin{proof}
 For finite threads the result follows by an induction on the starting point of the thread, where the base cases (the thread starting at the root) is proved by a subinduction on the length of the thread. %
  The result for infinite threads follows from the result for finite ones.
\end{proof}
\begin{corollary}\label{cor:critical-height-and-projection}
  Let \(\pi \in \nwProof\) and \(f\) be a \(*\)-assignment for \(S^\pi\), and \(b \in \branches(\Projection^f(\pi))\).
  Then
  \begin{enumerate}
    \item \(\progressPoints(\pi,b) \subseteq \progressPoints(\Projection^f(\pi),b)\).
    \item \(\criticalHeight(\Projection^f(\pi),b) \leq \criticalHeight(\pi,b)\).
    \item \(\Projection^f(\pi) \in \nwProof\).
  \end{enumerate}
\end{corollary}
\begin{proof}
  Proof of 1)
  Let \(i \in \progressPoints(\pi, b)\), so there is a \(t\) which is an
  infinite \(*\)-thread principal infinitely often in \((\pi, b)\) such that
  \(t\) is principal at step \(i\) \((\pi,b)\), i.e.\
  \(t : [i, + \infty) \longrightarrow \NN \cup \set{-1}\) and
  \(t_i\) is principal at node \(b_i\) of \(\pi\).
  By Lemma~\ref{lm:projection-and-threads}, \(t\) is also an infinite thread of \((\Projection^f(\pi), b)\) which is principal infinitely often .
  In addition, by Corollary~\ref{cor:progressing-distinct-from-all-principal},
  we get that \(t\) does not agree with any \((\pi,b)\)-progressing \(*\)-thread \(t'\) starting at \(0\) with \(f(t'_0) \neq \bot\); so
  by Lemma~\ref{lm:projection-and-threads} we get that \(t\) denotes the
  same formulas in \((\Projection^f(\pi), b)\) so it is also a \(*\)-thread.

  But \(t\) being principal at step \(i\) in \((\pi,b)\) implies by Lemma~\ref{lm:projection-and-threads} that \(t\) is also principal
  at step \(i\) in \((\Projection^f(\pi),b)\) so \(i \in \progressPoints(\Projection^f(\pi), b)\).

  Proof of 2) It follows from 1) by the definition of \(\criticalHeight\).

  Proof of 3) If follows from 1) since a branch is progressing iff it has progressing points.
\end{proof}

Let \(\pi \in \text{p}\nwProof\) and \(f\) be a \(*\)-assignment for \(S^\pi\).
If \(f(k) = n \in \NN\) and for each \(j \neq i\) we have \(f(j) = \bot\),
then we will denote \(\Projection^f(\pi)\) as \(\Projection^n_k(\pi)\).

\begin{definition}
  We define the function \(\Om\colon \nwProof \longrightarrow \text{p}\wProof\)
  corecursively as:
  \[
    \Om\left(
    \AxiomC{\(\pi_0\)}
    \noLine
    \UnaryInfC{\(S_0\)}
    \AxiomC{\(\ldots\)}
    \AxiomC{\(\pi_{k-1}\)}
    \noLine
    \UnaryInfC{\(S_{k-1}\)}
    \RightLabel{\(r\)}
    \TrinaryInfC{\(S\)}
    \DisplayProof
    \right) =
    \AxiomC{\(\Om(\pi_0)\)}
    \noLine
    \UnaryInfC{\(S_0\)}
    \AxiomC{\(\ldots\)}
    \AxiomC{\(\Om(\pi_{k-1})\)}
    \noLine
    \UnaryInfC{\(S_{k-1}\)}
    \RightLabel{\(r\)}
    \TrinaryInfC{\(S\)}
    \DisplayProof    
  \]
  if \(r \neq *\text{L}\), and
  \[
      \hspace{-0.4cm}
    \Om\left(
      \AxiomC{\(\pi_0\)}
      \noLine
      \UnaryInfC{\(\Pi_l, \Pi_r \Rightarrow \beta\)}
      \AxiomC{\(\pi_1\)}
      \noLine
      \UnaryInfC{\(\Pi_l, \alpha, \alpha^* \Pi_r \Rightarrow \beta\)}
      \RightLabel{\(*\)L}
      \BinaryInfC{\(\Pi_l, \alpha^*, \Pi_r \Rightarrow \beta\)}
      \DisplayProof
    \right)
    =
  \]

  \[
    \AxiomC{\(\Om(\pi_0)\)}
    \noLine
    \UnaryInfC{\(\Pi_l, \Pi_r \Rightarrow \beta\)}
    \AxiomC{\(\Om(\Projection^{i}_{k+1}(\pi_1))\)}
    \noLine
    \UnaryInfC{\((\Pi_l, \alpha, \alpha^i, \Pi_r \Rightarrow \beta)_{i \in \mathbb{N}}\)}
    \RightLabel{\(*\text{L}_\omega\)}
    \BinaryInfC{\(\Pi_l, \alpha^*, \Pi_r \Rightarrow \beta\)}
    \DisplayProof
  \]
  where $\alpha^*$ in the conclusion is at position $k$.
  Also, given a \(\pi \in \text{p}\wProof\) and \(b \in \branches(\pi)\),
  we define \(\path(\pi, b) \in \omega^\omega\) as:
  \[
    \path(\pi,b)_i = \begin{cases}
      0 &\text{if } R_i = *\text{L}_\omega \text{ and } b_i = 0, \\
      1 &\text{if } R_i = *\text{L}_\omega \text{ and } b_i > 0, \\
      b_i &\text{otherwise}.
    \end{cases}
  \]
  where \((S_i, R_i)_{i \in \mathbb{N}}\) is the sequence of sequents and rules of \((\pi,b)\).
\end{definition}

\(\path(\pi,b)\) takes a branch \(b\) of \(\pi\) and returns a sequence immitating \(b\),
but when \(b\) goes through a premise the \(\omega\)-rule distinct from the first one \(\path(\pi,b)\) goes to the second premise.
So if \(\pi\) is \(\Om(\tau)\) for some \(\tau \in \nwProof\), then \(\path(\pi,b) \in \branches(\tau)\), i.e., when applied to a branch in a proof that comes from \(\Om\) we obtain a branch in the original non-wellfounded proof.

We also have that going up one step in the branch \(\path(\pi,b)\) is the same as going to the \(b_0\)-th premise in \(\pi\) and then applying \(\path\), i.e., 
\[
  \tail(\path(\pi, b)) = \path(\tau, \tail(b))
\]
where \(\tau\) is the subtree of \(\pi\) generated by the node \(b_0\).
We use these facts in the following lemma.

\begin{lemma}
  Let \(\pi \in \nwProof\), then \(\Om(\pi) \in \wProof\) with the same conclusion.
  As a consequence,
  \(\gentzenContinuousActionLattice + \mathcal{R} \vdash S\) implies
  \(\gentzenContinuousActionLatticeOmega + \mathcal{R} \vdash S\).
\end{lemma}
\begin{proof}
  We have to prove that \(\branches(\Om(\pi)) = \varnothing\). So,
  assume \(b' \in \branches(\Om(\pi))\) and define \( b = \path(\Om(\pi), b')\), since \(b \in \branches(\pi)\)
  we have that \((\pi, b)\) is a progressing branch.
  We proceed by induction in the critical height of \((\pi, b)\) (which
  is a natural number as this branch is progressing).

  Case \(\criticalHeight(\pi, b) = 0\).
  Then we have that \(\pi\) and \(\Om(\pi)\) are respectively of shape
  \[ 
    \AxiomC{\(\pi_0\)}
    \noLine
    \UnaryInfC{\(\Pi_l, \Pi_r \Rightarrow \beta\)}
    \AxiomC{\(\pi_1\)}
    \noLine
    \UnaryInfC{\(\Pi_l,\alpha,\alpha^*, \Pi_r \Rightarrow \beta\)}
    \RightLabel{\(*\)L}
    \BinaryInfC{\(\Pi_l, \alpha^*, \Pi_r \Rightarrow \beta\)}
    \DisplayProof
  \]
  \[
    \AxiomC{\(\Om(\pi_0)\)}
    \noLine
    \UnaryInfC{\(\Pi_l, \Pi_r \Rightarrow \beta\)}
    \AxiomC{\(\Om(\Projection_{k+1}^i(\pi_1))\)}
    \noLine
    \UnaryInfC{\((\Pi_l, \alpha, \alpha^i, \Pi_r \Rightarrow \beta)_{i \in \mathbb{N}}\)}
    \RightLabel{\(*\text{L}_\omega\)}
    \BinaryInfC{\(\Pi_l, \alpha^*, \Pi_r \Rightarrow \beta\)}
    \DisplayProof
  \]
  and, since \(\criticalHeight(\pi,b) = 0\) i.e.\ \(0 \in \progressPoints(\pi,b)\),
  there is an  infinite \(*\)-thread \(t\) principal infinitely often in
  \((\pi, b)\) that starts at the displayed \(\alpha^*\) and is at position \(k\).
  This implies that \(b = 1 : \tail(b)\), i.e.\ the branch must go to the
  right,
  yielding that \(b' = (n + 1) : \tail(b')\) for some \(n \in \NN\), so
  \(\tail(b') \in \branches(\Om(\Projection_{k+1}^i(\pi_1)))\).
  Also \(\tail(b) = \tail(\path(\Om(\pi), b')) =
  \path(\Om(\Projection_{k+1}^n(\pi_1)), \tail(b'))\).
  This means that 
  \(\tail(b) \in \Projection_{k+1}^n(\pi_1)\),
  so using the contrapositive of Lemma~\ref{lm:star-progressing-thread-implies-stucked} we get that no \(*\)-thread of \((\pi,b)\) starting at position \(k + 1\) (i.e.\ in the displayed \(\alpha^*\) of \(\pi_1\)) is progressing, or equivalently, principal infinitely often, since it is a \(*\)-thread.
  We have contradicted that \(t\) is an infinite \(*\)-thread principal infinitely often in \((\pi, b)\), as desired.

  Case \(\criticalHeight(\pi, b) > 0\).
  There are two subcases.

  Subcase 1: the last rule of \(\pi\) is \(*\text{L}\).
  Then we have that \(\pi\) and \(\Om(\pi)\) are respectively of shape
  \[ 
    \AxiomC{\(\pi_0\)}
    \noLine
    \UnaryInfC{\(\Pi_l, \Pi_r \Rightarrow \beta\)}
    \AxiomC{\(\pi_1\)}
    \noLine
    \UnaryInfC{\(\Pi_l,\alpha,\alpha^*, \Pi_r \Rightarrow \beta\)}
    \RightLabel{\(*\)L}
    \BinaryInfC{\(\Pi_l, \alpha^*, \Pi_r \Rightarrow \beta\)}
    \DisplayProof
  \]
  \[
    \AxiomC{\(\Om(\pi_0)\)}
    \noLine
    \UnaryInfC{\(\Pi_l, \Pi_r \Rightarrow \beta\)}
    \AxiomC{\({(\Om(\Projection_{k+1}^i(\pi_1)))}_{i \in \NN}\)}
    \noLine
    \UnaryInfC{\(\Pi_l, \alpha, \alpha^i, \Pi_r \Rightarrow \beta\)}
    \RightLabel{\(*\text{L}_\omega\)}
    \BinaryInfC{\(\Pi_l, \alpha^*, \Pi_r \Rightarrow \beta\)}
    \DisplayProof
  \]
  If \(b = 0 : \tail(b)\), then  \(\tail(b) \in \branches(\pi_0)\).
  Since \(\criticalHeight(\pi, b) > 0\) we get
  \(\criticalHeight(\pi_0, \tail(b)) < \criticalHeight(\pi, b)\), so we can
  use the induction hypothesis.

  Finally, assume \(b = 1 : \tail(b)\), so \(b' = (n+1) : \tail(b')\) and, 
\[
    \tail(b) = \path(\Om(\Projection_{k+1}^n(\pi_1)), \tail(b'))
    \in \branches(\Projection_{k+1}^n(\pi_1)) \subseteq \branches(\pi_1).
\]
  Since \(\criticalHeight(\pi, b) > 0\) we know that
  \(\criticalHeight(\pi_1, \tail(b)) < \criticalHeight(\pi, b)\).
  By Corollary~\ref{cor:critical-height-and-projection} we get
  \(\criticalHeight(\Projection_{k+1}^{n}(\pi_1), \tail(b)) <
  \criticalHeight(\pi, b)\).
  So we can use the induction hypothesis (with \(\Projection_{k+1}^{n}(\pi_1)\)
  instead of \(\pi_1\)).

  Subcase 2: last rule of \(\pi\) is not \(*\text{L}\).
  Straightforward from the induction hypothesis.
\end{proof}

\begin{theorem}\label{t:equiv-of-systems}
  Let \(\mathcal{R}\) be any set of linear rules.
  Then \(\gentzenContinuousActionLatticeOmega + \mathcal{R} \vdash S\) iff \(\gentzenContinuousActionLattice + \mathcal{R} \vdash S\).
\end{theorem}

\section{Cut-free completeness}\label{sec:cut-free completeness}

In this section, we deploy tools from algebraic proof theory to obtain proofs of cut-free completeness of proof systems extending \(\gentzenContinuousActionLatticeOmega\) in a modular and transparent fashion. Broadly speaking, algebraic proof theory applies \emph{residuated frames} \cite{GJ2013} in order to link (wellfounded) proof systems to certain algebraic models with an underlying complete lattice reduct, allowing for semantic proofs of cut elimination among other things. Algebraic proof theory has been used extensively to provide modular analytic proof theory for an expansive class of logical systems in the vicinity of the Full Lambek calculus; see, e.g., \cite{CGT2012,CGT2017}. Action algebras amount to extensions of residuated lattices (the algebraic models of the Full Lambek calculus) by the Kleene star, and this work extends the algebraic proof theory regime to account for Kleene star.

\subsection{Residuated frames}
A \emph{residuated frame} is a structure of the form $\m{W} = (W, W', N, \circ, \epsilon)$ such that 
\begin{itemize}
\item $W$ and $W'$ are sets and $N \subseteq W \times W'$,
\item $(W,\circ, \epsilon)$ is a monoid, and
\item for all $x,y \in W$ and $z \in W'$ there exist elements $x \Lresidual z, z \Rresidual y \in W'$ such that
\[
x \circ y \mathrel{N} z \iff y \mathrel{N} x \Lresidual z \iff x \mathrel{N} z \Rresidual y.
\]
\end{itemize}
A binary relation that satisfies the last property is called \emph{nuclear}.

Let $\m{W}$ be a residuated frame. We define the two operations  \({\cdot}^{\rhd} \colon \wp(W) \longrightarrow \wp(W')\) and \({\cdot}^{\lhd} \colon \wp(W') \longrightarrow \wp(W)\) by
\begin{align*}
	 X^{\rhd} &:= \set{y \in W' \mid \forall x\in X.\ x \mathrel{N} y}, \\
	 Y^{\lhd} &:= \set{x \in W \mid \forall y\in Y.\ x \mathrel{N} y}.
\end{align*}
Together, they form a Galois connection $({^\lhd},{^\rhd})$; we define a closure operator \(\gamma \colon \wp(W) \longrightarrow \wp(W)\), \(\gamma(X) := X^{\rhd \lhd}\). Since $N$ is nuclear, it can be shown that \(\gamma\) is a nucleus, i.e.,  for any \(X,Y \subseteq W\)
\[ 
	\gamma(X) \circ \gamma(Y) \subseteq \gamma(X \circ Y),
\]
where \(X \circ Y := \set{x \circ y \mid x \in X, y \in Y}\). We refer to  \cite{GJ2013} for proofs of the aforementioned facts. 

The \emph{dual algebra} of a residuated frame $\m{W}$ is the algebra
\[
\m{W}^+ = (W^+, \cap, \cup_\gamma, \circ_\gamma, \lresidual,\rresidual, \gamma(\{\epsilon\})),
\]
where \(W^+ = \gamma[\wp(W)]\) and the operations are defined by
\begin{align*}
  X \cup_\gamma Y &:= \gamma(X \cup Y), \\
  X \circ_\gamma Y &:= \gamma(X \circ Y),\\
  X \lresidual Y &:= \set{y \mid X \circ \set{y} \subseteq Y}, \\
  Y \rresidual X &:= \set{y \mid \set{y} \circ X  \subseteq Y}.
\end{align*}
\begin{lemma}[\cite{GJ2013}]\label{l:dual-alg-reslat}
If $\m{W}$ is a residuated frame, then $\m{W}^+$ is a residuated lattice.
\end{lemma}

\begin{definition}
A \emph{Gentzen frame} is a pair $(\m{W},\m{A})$ where 
\begin{itemize}
\item $\m{W} = (W, W', N, \circ, \epsilon )$ is a residuated frame. 
\item $\m{A}$ is an algebra in the language $\{\meet,\join,\cdot, \under,\ovr, 1\}$ with an injection from $A$ to $W$ under which we identify $A$ as a subset of $W$.
\item $(W,\circ,\epsilon)$ is generated by $A$.
\item there is an injection from $A$ into $W'$ whose image we will identify with $A$.
\item $N$ satisfies the rules in Figure~\ref{fig:gentzen} for, $x,y\in W$, $z\in W'$,  $a,a_1,a_2,b \in A$.
\end{itemize} 
A \emph{cut-free Gentzen frame} is defined analogously, but without requiring that it satisfies (Cut).
\end{definition}

\begin{figure}
	\[
	\AxiomC{\(x \mathrel{N} a\)}
	\AxiomC{\( a\mathrel{N} z  \)}
	\RightLabel{(Cut)}
	\BinaryInfC{\( x \mathrel{N} z \)}
	\DisplayProof
	\hspace{1cm}
	\AxiomC{}
	\RightLabel{(Id)}
	\UnaryInfC{\(a \mathrel{N} a\)}
	\DisplayProof
	\]
	\[
	\AxiomC{\(x \mathrel{N} a \)}
	\AxiomC{\(b \mathrel{N} z \)}
	\RightLabel{($\under$L)}
	\BinaryInfC{\(a\under b \mathrel{N} x \Lresidual z \)}
	\DisplayProof
	\hspace{1cm}
	\AxiomC{\(x \mathrel{N} a \Lresidual b \)}
	\RightLabel{($\under$R)}
	\UnaryInfC{\(x \mathrel{N} a \under b \)}
	\DisplayProof
	\]
	\[
	\AxiomC{\(x \mathrel{N} a \)}
	\AxiomC{\(b \mathrel{N} z \)}
	\RightLabel{($\ovr$L)}
	\BinaryInfC{\(b \ovr a \mathrel{N} z \Rresidual x \)}
	\DisplayProof
	\hspace{1cm}
	\AxiomC{\(x \mathrel{N} b \Rresidual a \)}
	\RightLabel{($\ovr$R)}
	\UnaryInfC{\(x \mathrel{N} b \ovr a \)}
	\DisplayProof
	\]
	\[
	\AxiomC{\(a\circ b \mathrel{N} z \)}
	\RightLabel{($\cdot$L)}
	\UnaryInfC{\(a\cdot b \mathrel{N} z \)}
	\DisplayProof
	\hspace{1cm}
	\AxiomC{\(x \mathrel{N} a \)}
	\AxiomC{\(y \mathrel{N} b \)}
	\RightLabel{($\cdot$R)}
	\BinaryInfC{\(x \circ y \mathrel{N}  a \cdot b \)}
	\DisplayProof
	\]
	\[
	\AxiomC{\(b \mathrel{N} z \)}
	\RightLabel{($\meet$L$_i$)}
	\UnaryInfC{\(a_0\meet a_1 \mathrel{N} z \)}
	\DisplayProof
  \hspace{0.5cm}
	\AxiomC{\(x \mathrel{N} a \)}
	\AxiomC{\(x \mathrel{N} b \)}
	\RightLabel{($\meet$R)}
	\BinaryInfC{\(x \mathrel{N}  a \meet b \)}
	\DisplayProof
	\]
	\[
	\AxiomC{\(a \mathrel{N} z \)}
	\AxiomC{\(b \mathrel{N} z \)}
	\RightLabel{($\join$L)}
	\BinaryInfC{\(a\join b \mathrel{N}  z \)}
	\DisplayProof
	\hspace{0.5cm}
	\AxiomC{\(x \mathrel{N} a_i \)}
	\RightLabel{($\join$R$_i$)}
	\UnaryInfC{\(x \mathrel{N} a_0\join a_1 \)}
	\DisplayProof
  \]
	\[
	\AxiomC{\(\epsilon \mathrel{N} z\)}
	\RightLabel{(1L)}
	\UnaryInfC{\( 1 \mathrel{N} z \)}
	\DisplayProof
	\hspace{1cm}
	\AxiomC{}
	\RightLabel{(1R)}
	\UnaryInfC{\(\epsilon \mathrel{N} 1\)}
	\DisplayProof
	\]
	
	\caption{Gentzen rules}
	\label{fig:gentzen}\end{figure}

\begin{definition}
A \emph{(cut-free) $*$-Gentzen frame} is pair $(\m{W},\m{A})$ such that $\m{A}$ is an algebra in the language $\{\meet,\join,\cdot, \under,\ovr,{}^*,1,0\}$ such that if $\m{A}'$ is the $\{{}^\ast,0 \}$-free reduct of $\m{A}$, then $(\m{W}, \m{A}')$ is a (cut-free) Gentzen frame and $N$ satisfies the additional rules in Figure~\ref{fig:star-gentzen} for $a\in A$,  $x,y\in W$, $z\in W'$, where for $x\in W$, $x^{(n)}$ is the $n$-th power with respect to $\circ$.
\end{definition}

\begin{figure}
	\[
	\AxiomC{\( x\mathrel{N} 0 \)}
	\RightLabel{($0$L)}
	\UnaryInfC{\( x  \mathrel{N} z \)}
	\DisplayProof
	\hspace{1cm}
	\AxiomC{\((a^{(n)} \mathrel{N} z)_{n\in \mathbb{N}}\)}
	\RightLabel{($*$L)}
	\UnaryInfC{\(a^* \mathrel{N} z\)}
	\DisplayProof
\]
\[
	\AxiomC{}
	\RightLabel{($*$R$_0$)}
	\UnaryInfC{\(\epsilon \mathrel{N} a^*\)}
	\DisplayProof
	\hspace{0.3cm}
	\AxiomC{\(x \mathrel{N} a\)}
	\AxiomC{\(y \mathrel{N} a^*\)}
	\RightLabel{($*$R$_1$)}
	\BinaryInfC{\(x \circ y \mathrel{N} a^*\)}
	\DisplayProof
	\]
	
	\caption{$*$-Gentzen rules}
	\label{fig:star-gentzen}\end{figure}

The \emph{dual algebra} of a (cut-free) $*$-Gentzen frame \((\m{W},\m{A})\) is the algebra 
	\[
	\dualAlgebra = (W^+, \cap, \cup_\gamma, \circ_\gamma, \lresidual,\rresidual, {^{*_\gamma}}, 0^{\lhd}, \gamma(\{\epsilon\})),
	\]
where we define
\[
  X^{*_\gamma} := \gamma(X^*)
\]
for
\[
  X^* = \set{x_0 \circ \cdots \circ x_{n-1} \mid x_i \in X, \, n\in \NN}.
\]
Note that we use the same notation as for the dual algebra of a residuated frame, but it will be clear from the context which dual algebra we mean.

\begin{proposition}
Let $(\m{W},\m{A})$ be a (cut-free) $*$-Gentzen frame. Then $\m{W}^+$ is a $*$-continuous action lattice.
\end{proposition}
\begin{proof}
Since $\m{W}$ is a residuated frame, it follows from Lemma~\ref{l:dual-alg-reslat} that $(W^+, \cap, \cup_\gamma, \circ_\gamma, \lresidual,\rresidual,  \gamma(\{\epsilon\}))$ is a complete residuated lattice. So it remains to show that $0^{\lhd}$ is the bottom element and $X^{*_\gamma} = \bigvee_{n\in \NN} X^n$ for each $X\in W^+$, where $X^n$ is the $n$-fold product of $X$ with respect to $\circ_\gamma$. For the first part let $X \in W^+$ and $x\in 0^{\lhd}$, i.e., $x\mathrel{N} 0$. It follows from (0L), that $x\mathrel{N} z$ for each $z \in X^{\rhd}$, yielding $x \in X^{\rhd\lhd}$.  Hence $0^{\lhd} \subseteq X$.  

For the second part let $X\in W^+$  and note that
\[
\bigcup_{n\in \NN} X^{(n)} =  X^*.
\]
Thus, since $\gamma$ is a nucleus, we get
\begin{align*}
\bigvee_{n\in \NN} X^n &= \bigvee_{n\in \NN} \gamma(X^{(n)}) \\
& = \gamma(\bigcup_{n\in \NN} X^{(n)}) \\
&= \gamma( X^*) =X^{*_\gamma}. \qedhere
\end{align*}
\end{proof}

Let $\m{A}$ and $\m{B}$ be two algebras in the language $\{\meet,\join,\cdot, \under,\ovr,{}^*,1,0\}$. A \emph{quasimorphism} from $\m{A}$ to $\m{B}$ is a map $f\colon A \to \wp(B)$ such that  $c_B \in F(c_A)$ for $c\in \{1,0\}$, $F(a)^{*_{B}} \subseteq F(a^{*_A})$, and $F(a) \bullet_B F(B) \subseteq F(a\bullet_A b)$ for $\bullet \in \{\meet,\join,\cdot,\under,\ovr\}$, where for $X,Y \in \wp(B)$, $X^{*_B} = \{x^{\ast_B} \mid x \in X \}$ and $X \bullet_{B} Y = \{x\bullet_B y \mid x\in X,y\in Y\}$.

\begin{lemma}[cf. \cite{GJ2013}]\label{lemma:quasi}
Let  $(\m{W},\m{A})$ be a cut-free $*$-Gentzen frame. 
\begin{enumerate}
\item The map $F\colon A \to \wp(W^+)$,
\begin{align*}
F(a) = \{X \in W^+ \mid a \in X \subseteq a^{\lhd}\}
\end{align*}
is a quasimorphism from $\m{A}$ to $\m{W}^+$.
\item If $(\m{W},\m{A})$ is a $*$-Gentzen frame, then $f\colon A \to W^+$, $f(a) = a^{\rhd\lhd} = a^{\lhd}$ is a homomorphism from $\m{A}$ to $\m{W}^+$. Moreover, if $N$ is antisymmetric, then $f$ is an embedding.
\end{enumerate}
\end{lemma}
\begin{proof}
1) The proof for the connectives $\{\meet,\join,\cdot, \under,\ovr,1\}$ can be found in \cite{GJ2013}. So we only need to consider the remaining connectives $0$ and ${}^*$.
First note that by definition $F(0) = \{X\in W^+ \mid 0 \in X \subseteq 0^\lhd \}$. So clearly $0^{\lhd} \in F(0)$. 

Now for the connective ${}^*$, let $a \in A$. Then for $X \in F(a)^{*_\gamma}$, there exists a $Y \in F(a)$, i.e., $a \in Y \subseteq a^{\lhd}$ such that $X = Y^{*_\gamma}$, or, equivalently, $a\in Y$ and for each $y\in Y$, $y \mathrel{N} a$. Thus for each $n\in \NN$, $a^{(n)} \in Y^*$, yielding $a^{(n)} \mathrel{N} z$ for each $z\in (Y^*)^{\rhd}$. So, by ($*$L), $a^* \mathrel{N} z$ for each $z\in (Y^*)^{\rhd}$, yielding $a^* \in (Y^{*})^{\rhd\lhd} = Y^{*_\gamma}  =X$. On the other hand, by ($*$R$_0$),  $\epsilon \mathrel{N} a^*$ and for each $x \in Y$, $x \mathrel{N} a$, so, by ($*$R$_1$), $x = x \circ \epsilon \mathrel{N} a^*$. Now inductively applying ($*$R$_1$) yields $x_0\circ x_1 \circ \dots \circ x_n \mathrel{N} a^*$ for all $x_0,x_1,\dots, x_n \in Y$, i.e, $Y^* \subseteq (a^*)^{\lhd}$ and thus $X = Y^{*_\gamma} \subseteq (a^*)^\lhd$. Hence, $X \in F(a^*)$ and we conclude that $F(a)^{*_\gamma} \subseteq F(a^*)$.
 
 For 2) see \cite{GJ2013}.
\end{proof}

\subsection{Cut-free completeness}

\begin{definition}
For a set of analytic rules $\mathcal{R}$ we define the frame $\m{W}_\mathcal{R} = (W,W',N_\mathcal{R}, \circ, \epsilon)$, where $\epsilon$ is the empty sequence,
\begin{enumerate}
\item \(W = {\term(\mathcal{L}_{\AL})}^*\), i.e., finite sequences of terms over the language of action lattices;
		\item \(W' = {\term(\mathcal{L}_{\AL})}^* \times {\term(\mathcal{L}_{\AL})}^* \times {\term(\mathcal{L}_{\AL})}\);
		\item \(\circ\) is composition of sequences;
		\item we define \(N_\mathcal{R} \subseteq W \times W'\) as
		\begin{align*}
      &\Gamma \mathrel{N_\mathcal{R}^\omega} (\Sigma_l, \Sigma_r, \alpha) \text{ iff } \\ 
      &\hspace{1cm}\gentzenContinuousActionLatticeOmega+\mathcal{R} \vdash \Sigma_l, \Gamma, \Sigma_r \Rightarrow \alpha. 
			\end{align*}
\end{enumerate}
\end{definition}

\begin{proposition}
For any set of analytic  rules $\mathcal{R}$, $(\m{W}_\mathcal{R},\termAlg(\mathcal{L}_{\AL}))$ is a cut-free $*$-Gentzen frame, where we identify $\alpha \in \term(\mathcal{L}_{\AL})$ with $(\epsilon,\epsilon,\alpha) \in W'$.
\end{proposition}

\begin{proof}
It is straightforward to check that the frame is a cut-free Gentzen frame when restricted to the signature without $0$ and ${^*}$, since the proof system $\gentzenContinuousActionLatticeOmega$ is  `essentially' an extension of the Full Lambek Calculus without falsum (see \cite{GJ2013}). Also, by Lemma~\ref{l:adm-of-frame-rules}, it follows that the frame satisfies ($0$L).  Moreover it is clear, by definition, that the frame satisfies the other  $*$-Gentzen rules, noting that they correspond to the obvious rules of the calculus.
\end{proof}

For a class of algebras $\mathsf{K}$ and a set of quasiequations $Q$ we denote by $\mathsf{K} + Q$ the class axiomatized relative to $\mathsf{K}$ by $Q$.  Let $q = (\alpha_1\leq \beta \mathbin{\&} \dots \mathbin{\&} \alpha_n\leq \beta \implies \alpha_0 \leq \beta)$ be an analytic quasiequation. So in particular $\alpha_i$ are products of variables (or $1$) and $\beta$ is a variable.   We say that a residuated frame $\m{W}$ satisfies $q$ if for each valuation $v$ into $(W,\circ, \epsilon)$,
\[
v(\alpha_1) \mathrel{N} v(\beta)\text{ and } \dots \text{ and } v(\alpha_n) \mathrel{N} v(\beta) \implies v(\alpha_0) \mathrel{N} v(\beta).
\]

\begin{theorem}[\cite{CGT2012}]\label{theorem:frame-to-dual-equation}
Let $\m{W}$ be a residuated frame. Then for each analytic quasiequation $q$, $\m{W}$ satisfies $q$ if and only if $\m{W}^+$ satisfies $q$.
\end{theorem}

\begin{corollary}\label{cor:dual-of-proof-system-satisfies-qe}
For any set of analytic rules $\mathcal{R}$, $(\m{W}_\mathcal{R})^+$ satisfies  $Q_a(\mathcal{R})$.
\end{corollary}

\begin{proof}
By definition, the frame $\m{W}_\mathcal{R}$ satisfies the analytic quasiequations in $Q_a(\mathcal{R})$, noting that containment of the rules in the proof system corresponds exactly to validity of the corresponding analytic quasiequation in the frame.  Now the claim follows from Theorem~\ref{theorem:frame-to-dual-equation}.
\end{proof}

For a sequent $S = (\alpha_0,\dots, \alpha_n \Rightarrow \beta)$, we define the equation $\varepsilon(S) := \alpha_0 \cdots \alpha_n \leq \beta$. 

\begin{corollary}
For any set of analytic rules $\mathcal{R}$, the system $\gentzenContinuousActionLatticeOmega+\mathcal{R}$ is complete with respect to the class of $^*$-continuous action lattices axiomatized by $Q_a(\mathcal{R})$, i.e., for each sequent $S$
\begin{multline*}
\AL^* + Q_a(\mathcal{R}) \models \varepsilon(S) \implies 
\gentzenContinuousActionLatticeOmega+\mathcal{R} \vdash S.
\end{multline*}
\end{corollary}
\begin{proof}
Let  $v \colon \term(\mathcal{L}_{\AL}) \to (\m{W}_\mathcal{R})^+$ be the homomorphism induced by $v(x) = \gamma(\{x\})$ and homomorphically extended to all terms. Then we have $x \in v(x) \subseteq x^{\lhd}$ for variables and applying Lemma~\ref{lemma:quasi}~1 inducitvely we get that $\alpha \in v(\alpha) \subseteq \alpha^{\lhd}$ for each term $\alpha \in  \term(\mathcal{L}_{\AL})$. Now let $\alpha_0,\dots, \alpha_n \Rightarrow \beta$  be a sequent such that
\(
\AL^* + Q_a(\mathcal{R}) \models \alpha_0\cdots \alpha_n \leq \beta
\).
Then, by Corollary~\ref{cor:dual-of-proof-system-satisfies-qe}, we have $(\m{W}_\mathcal{R})^+ \models  \alpha_0\cdots \alpha_n \leq \beta$, yielding 
\[
\alpha_0 \circ \dots \circ \alpha_n  \in v(\alpha_0)\circ_\gamma \cdots \circ_\gamma v(\alpha_n) \subseteq v(\beta) \subseteq \beta^{\lhd}.
\]
Hence $\gentzenContinuousActionLatticeOmega + \mathcal{R} \vdash \alpha_0,\dots, \alpha_n \Rightarrow \beta$.
\end{proof}
Soundness and Theorem~\ref{t:equiv-of-systems} yields:
\begin{theorem}\label{thm:main}
Let $\mathcal{R}$ be a set of analytic rules. Then  $\gentzenContinuousActionLatticeOmega+\mathcal{R}$  and $\gentzenContinuousActionLattice+\mathcal{R}$  are sound and complete with respect to the class of $*$-continuous action lattices axiomatized by $Q_a(\mathcal{R})$.
\end{theorem}
In particular, since Cut is sound for any $*$-continuous action lattice, we obtain analyticity:

\begin{corollary}
Let $\mathcal{R}$ be a set of analytic rules. Then the systems  $\gentzenContinuousActionLatticeOmega+\mathcal{R}$  and $\gentzenContinuousActionLattice+\mathcal{R}$ have cut-elimination.
\end{corollary}

\subsection{MacNeille completions}
There is an emerging literature on the connection between the analyticity of proof systems and the closure of classes of corresponding algebraic models under various kinds of completions of ordered algebraic structures (see \cite{CGPT2022}), such as canonical completions (see \cite{FGGM2021}). The techniques utilized here---using residuated frames---connect analyticity to closure under MacNeille completion, and as an application of the previous results we obtain closure under MacNeille completions for classes of   $*$-continuous action lattices axiomatized relative to the class $\AL^*$ by  analytic quasiequations.

\begin{theorem}\label{cor:MacNeille}
Let $\mathsf{K}$ be the class of $*$-continuous action lattices axiomatized relative to the class $\AL^*$ by a set of analytic quasiequations $Q$. Then $\mathsf{K}$ is closed under MacNeille completions.
\end{theorem}

\begin{proof}
For each continuous action algebra $\m{A} \in \mathsf{K}$ consider the frame $\m{W}_\m{A} = (A,A,\leq, \cdot, 1,0 )$, then the  $\{{}^*,0\}$-free reduct of $(\m{W}_{A},\m{A})$ is a Gentzen frame, noting that $W = W' = A$ (see \cite{GJ2013}). It is also a $*$-Gentzen frame, since for all $a,b,c\in A$,
\begin{align*}
a\leq 0\implies a = 0 \leq b,\\
a^n \leq b \text{ for all } n\in \NN \implies a^* = \bigvee_{n\in \NN} a^n \leq b, \\
1 \leq a^*, \\
b\leq a \text{ and } c \leq a^* \implies bc \leq aa^* = a^*.
\end{align*}
Thus, by Lemma~\ref{lemma:quasi},  the map $f\colon \m{A} \to (\m{W}_\m{A})^+$, $f(a) = a^{\lhd}$ is an embedding. Moreover, $(\m{W}_\m{A})^+$ is the MacNeille completion of $\m{A}$ (see \cite{GJ2013}) and clearly  $\m{A}$ satisfies an analytic quasiequation if and only if  $\m{W}_\m{A}$ satisfies it. So $\m{W}_{A}$ satisfies every quasiequation in $Q$ and, by Theorem~\ref{theorem:frame-to-dual-equation}, $(\m{W}_\m{A})^+$ satisfies every quasiequation in $Q$ and thus $(\m{W}_\m{A})^+ \in \mathsf{K}$.
\end{proof}

\section{Conclusion}

In this paper, we have illustrated a uniform method for the producing cut-free proof systems for classes of $*$-continuous action lattices that are defined relative to all $*$-continuous action lattices by analytic quasiequations. The production of such cut-free proof systems is a necessary precondition for realizing the benefits typical of proof-theoretic methods, such as the effective computation of interpolants. The latter has recently been carried out in the context of certain cyclic proof systems (see e.g. \cite{ALM2021}), and the present work could form the basis for extending methods for interpolation to non-wellfounded proof systems.

An important aspect of the latter is that we have not, in the present work, carried out our investigation solely in the environment of non-wellfounded proof systems. While the algebraic methods employing residuated frames apply without modification to the non-wellfounded environment, the soundness of the non-wellfounded proof systems with respect to the algebraic semantics discussed in this paper have relied on a translation to wellfounded systems. We do not know whether it is possible to address soundness without resorting to such a translation.

The natural next step in extending the present work is to drop the requirement of $*$-continuity. This is not an easy matter: Semantic methods based on lattice completions, as the residuated frame method, apparently demand some delicate modification to apply beyond the $*$-continuous case, if such methods may be applied at all. Successfully reproducing these methods in beyond the $*$-continuous context would represent a significant milestone both for proof-theoretic treatments of Kleene algebra as well as in algebraic proof theory.

\bibliographystyle{plain}

\end{document}